\documentclass[letterpaper]{IEEEtran}%
\usepackage{amsfonts}
\usepackage{amsmath}
\usepackage{amssymb}
\usepackage{amsthm}
\usepackage{graphicx}
\usepackage[colorlinks=true,linkcolor=blue,citecolor=red,plainpages=false,pdfpagelabels]%
{hyperref}

\usepackage{color}

\providecommand{\U}[1]{\protect\rule{.1in}{.1in}}
\newtheorem{theorem}{Theorem}

\newtheorem{corollary}{Corollary}

\newtheorem{definition}{Definition}

\newtheorem{lemma}{Lemma}

\newtheorem{proposition}{Proposition}
\newtheorem{remark}{Remark}

\allowdisplaybreaks

\def\Tr{\operatorname{Tr}}

\def\supp{\operatorname{supp}}

\def\>{\rangle}
\def\<{\langle}
\def\({\left(}
\def\){\right)}
\def\[{\left[}
\def\]{\right]}
\def\V{\Vert}

\def\id{\operatorname{id}}

\newcommand{\mc}[1]{\mathcal{#1}}
\newcommand{\wt}[1]{\widetilde{#1}}
\numberwithin{equation}{section}
\begin{document}
%\preprint{ }
\title{Quantum reading capacity:\\General definition and bounds}
\author{Siddhartha Das and
Mark M. Wilde
\thanks{Siddhartha Das was affiliated with the Hearne Institute for Theoretical Physics, Department of Physics and Astronomy,
Louisiana State University, Baton Rouge, Louisiana 70803, USA when this research was conducted and completed. He is now with the Centre for Quantum Information \& Communication (QuIC), \'Ecole polytechnique de Bruxelles, Universit\'e libre 
de Bruxelles, Brussels, B-1050, Belgium. Mark M.~Wilde is affiliated with the Hearne Institute for Theoretical Physics, Department of Physics and Astronomy, Center for Computation and Technology,
Louisiana State University, Baton Rouge, Louisiana 70803,  USA. Emails: sdas21@lsu.edu, mwilde@lsu.edu. This paper was presented in part at the Beyond i.i.d. in Information Theory Conference in Singapore during July 2017.}}

%\keywords{quantum reading, memory cell, channel discrimination, adaptive strategy, quantum reading capacity, covariant channels, teleportation-simulable, environment-parametrized, zero-error capacity, strong converse, second-order asymptotics}
%\pacs{}

\date{\today}
%\startpage{1}
%\endpage{10}
\maketitle

\begin{abstract}
Quantum reading refers to the task of reading out classical information stored in a read-only memory device. In any such protocol, the transmitter and receiver are in the same physical location, and the goal of such a protocol is to use these devices (modeled by independent quantum channels), coupled with a quantum strategy, to read out as much information as possible from a memory device, such as a CD or DVD. As a consequence of the physical setup of quantum reading, the most natural and general definition for quantum reading capacity should allow for an adaptive operation after each call to the channel, and this is how we define quantum reading capacity in this paper. We also establish several bounds on quantum reading capacity, and we introduce an environment-parametrized memory cell with associated environment states, delivering second-order and strong converse bounds for its quantum reading capacity. We calculate the quantum reading capacities for some exemplary memory cells, including a thermal memory cell, a qudit erasure memory cell, and a qudit depolarizing memory cell. We finally provide an explicit example to illustrate the advantage of using an adaptive strategy in the context of zero-error quantum reading capacity.
  
\end{abstract}
%\volumeyear{ }
%\volumenumber{ }
%\issuenumber{ }
%\eid{ }

\section{Introduction}
One of the primary goals of quantum information theory is to identify limitations on information processing when constrained by the laws of quantum mechanics. In general, quantum  information theory uses tools that are universally applicable to the processing of arbitrary quantum systems, which include quantum optical systems, superconducting systems, trapped ions, etc.~\cite{NC00}. The abstract approach to quantum information allows us to explore how to use the principles of quantum mechanics  for communication or computation tasks, some of which would not be possible without quantum mechanics. 

In \cite{BRV00}, a communication protocol was introduced in which a classical message is encoded in a set of unitary operations, and later on, one can read out the information stored in the unitary operations by calling them. Over a decade after \cite{BRV00} was published, this communication model was generalized and studied under the name \textquotedblleft quantum reading\textquotedblright\ in \cite{Pir11}, and it was  applied to the setting of an optical read-only memory. An optical read-only memory is one of the prototypical examples of quantum reading, and for this reason, quantum reading has been mainly considered in the context of optical realizations like CD-ROMs and DVDs (see \cite{LP16} and references therein). In this case, classical bits are encoded in the reflectivity and phase of memory cells, which can be modeled as a collection of pure-loss bosonic channels. More generally and abstractly, a memory cell is a collection of quantum channels, from which an encoder can select to form codewords for the encoding of a classical message (see Section~\ref{sec:notation} for a formal definition). Each quantum channel in a codeword, representing one part of the stored information, is read  only once.  In subsequent works \cite{PLG+11,LP16}, the model was extended to a memory cell consisting of arbitrary quantum channels.  In a quantum reading strategy, one exploits entangled states and collective measurements to help read out a classical message  stored in a read-only memory device. In many cases, one can achieve performance better than what can be achieved when using a classical strategy \cite{Pir11}. 

Some early developments in quantum reading  were based on a direct application of developments in quantum channel discrimination \cite{Kit97,DPP01,Aci01,WY06,TEGGLMPS08} (see also \cite{DFY09,CMW14,HHLW10,DGLL16}). 
However, the past few years have seen some progress in quantum reading: there have been developments in defining protocols for quantum reading (including limited definitions of reading capacity and zero-error reading capacity), giving upper bounds on the rates for classical information readout,  achievable rates for memory cells consisting of a particular class of bosonic channels, and details of a quantum measurement that can achieve non-trivial rates for memory cells consisting of a certain class of bosonic channels \cite{Pir11,PLG+11,GDN+11,GW12,WGTL12,GS13,LP16}.  Most recently, a task for secure reading of a memory device against a passive eavesdropper, called \textit{private reading}, has been introduced in \cite{DBW17}. The information-theoretic study of quantum reading is based on considerations coming from quantum Shannon theory \cite{W15book}, and the most abstract and general way to define the encoding of a classical message in a quantum reading protocol is as mentioned above, a sequence of quantum channels chosen from a given memory cell.

Hitherto, all prior works on quantum reading considered decoding protocols of the following form: A reader possessing a transmitter system entangled with an idler system sends the transmitter system through the coded sequence of quantum channels. Finally, the reader decodes the message by performing a collective measurement on the joint state of the output system and the idler system. 

However, the above approach neglects an important consideration: \textit{in a quantum reading protocol, the transmitter and receiver are in the same physical location}. We can thus refer to both devices as a single device called a transceiver. As a consequence of this physical setup, the most  general  and natural definition for quantum reading capacity should allow for the transceiver to perform an adaptive operation after each call to the memory, and this is how we define quantum reading capacity in this paper (see Section~\ref{sec:q-r-p}, as well as Figure~\ref{fig:reading-protocol} for a depiction of our modified definition of a quantum reading protocol).

In general, an adaptive strategy can have a significant advantage over a non-adaptive strategy in the context of quantum channel discrimination \cite{HHLW10}. Furthermore, a quantum channel discrimination protocol employing a non-adaptive strategy is a special case of one that uses an adaptive strategy. Since quantum reading bears close connections to quantum channel discrimination, we should suspect that adaptive operations could help to increase  quantum reading capacity in some cases, and this is one contribution of the present paper.

 We stress that the physical setup of quantum reading is rather different from that considered in a typical communication problem (see also \cite{DBW17} for a detailed discussion), in which the sender and receiver are in different physical locations. In this latter case, allowing for adaptive operations represents a different physical model and  is thus considered as a different kind of capacity, typically called a feedback-assisted capacity. However, as advocated above, the physical setup of quantum reading necessitates that there should be no such distinction between capacities: the quantum reading capacity should be defined as it is here, in such a way as to allow for adaptive operations.

Another point of concern with prior work on quantum reading is as follows: so far, all bounds on the quantum reading rate have been derived in the usual setting of quantum Shannon theory, in which the number of uses of the channels tends to infinity (also called the i.i.d.~setting, where i.i.d.~stands for ``independent and identically distributed'').
However, it is important for practical purposes to determine rates for quantum reading in the non-asymptotic scenario, i.e., for a finite number of quantum channel uses  and a given error probability for decoding. The information-theoretic analysis in the non-asymptotic case is motivated by the fact that in practical scenarios, we have only finite resources at our disposal \cite{RennerThesis,DR09,T15book}.

In this paper, we address some of the concerns mentioned above by giving the most general and natural definition for a quantum reading protocol and quantum reading capacity. We also establish bounds on the rates of quantum reading for wider classes of memory cells in both the asymptotic and non-asymptotic cases. First, we define a quantum reading protocol and quantum reading capacity in the most general setting possible by allowing for adaptive strategies. We give weak-converse, single-letter bounds on the rates of quantum reading protocols that employ either adaptive or non-adaptive strategies for arbitrary memory cells. We also introduce a particular kind of memory cell, which we call an
\textit{environment-parametrized memory cell} with associated environment states (see Section~\ref{sec:channel-symmetry} for definitions), for which stronger statements can be made for the rates and capacities in the non-asymptotic situation of a finite number of uses of the channels. We note that a particular kind of environment-parametrized memory cell consists of a collection of channels that are jointly teleportation simulable with associated resource states (see Definition~\ref{def:tel-cell} and \cite{BDSW96,HHH99} for  teleportation simulation, as well as \cite{NFC09,Mul12}). Many channels of interest obey these symmetries: some examples are erasure,  dephasing, thermal, noisy amplifier, and Pauli channels \cite{BDSW96,DP05,JWD+08,Mul12,PLOB15,WTB16,TW16}. Here we determine strong converse and  second-order bounds on the quantum reading capacities of environment-parametrized memory cells. Note that a strong converse rate~$R$ is such that the success probability of a sequence of protocols tends to zero with the number of channel uses if the actual rate of the protocols in the sequence exceeds the rate~$R$. Based on an example from \cite[Section 3]{HHLW10}, we show in  Section~\ref{sec:zero-error} that there exists a memory cell for which its zero-error reading capacity with adaptive operations is at least $\tfrac12$, but its zero-error reading capacity without adaptive operations is equal to zero. This example emphasizes how reading capacity should be defined in such a way as to allow for adaptive operations, as stressed in our paper.

The organization of our paper is as follows. In the next section, we begin by introducing standard notation, definitions, and necessary lemmas. We introduce two of the aforementioned classes of memory cells in Section~\ref{sec:channel-symmetry}. In Section~\ref{sec:q-r-p}, we give the most general and natural definition of a quantum reading protocol and quantum reading capacity. Section~\ref{sec:converse bounds} contains our main results, which were briefly summarized in the previous paragraph. In Section~\ref{sec:example}, we calculate quantum reading capacities for a thermal memory cell and for a class of jointly covariant memory cells, including a qudit erasure memory cell and a qudit depolarizing memory cell. In Section~\ref{sec:zero-error}, we provide an example to illustrate the advantage of adaptive operations over non-adaptive operations in the context of zero-error quantum reading capacity.  In the final section of the paper, we conclude and shed some light on possible future work.

\section{Preliminaries}

\label{sec:notation}

We begin by summarizing some of the standard notation, definitions, and lemmas that are used in the subsequent sections of the paper. 
 
 \subsection{Quantum states, measurements, channels, and memory cells, etc.}
 
Let $\mc{L}(\mc{H})$ denote the 
algebra of bounded linear operators acting on a Hilbert space $\mc{H}$. The
subset of $\mc{L}(\mc{H})$ 
containing all positive semi-definite operators is denoted by $\mc{L}_+(\mc{H})$. We denote the identity operator as $I$ and the identity superoperator as $\id$. The Hilbert space 
of a quantum system~$B$ is denoted by $\mc{H}_B$.
The state of a quantum system $B$ is represented by a density operator~$\rho_B$, which is a positive semi-definite operator with unit trace.
Let $\mc{D}(\mc{H}_B)$ denote the set of all elements $\rho_B\in \mc{L}_+(\mc{H}_B)$ such that $\Tr\{\rho_B\}=1$. The Hilbert space for a joint system $RB$ is denoted as $\mc{H}_{RB}$ where $\mc{H}_{RB}=\mc{H}_R\otimes\mc{H}_B$. The density operator of a joint system $RB$ is defined as $\rho_{RB}\in \mc{D}(\mc{H}_{RB})$, and the partial trace over $B$ gives the reduced density operator for system $R$, i.e., $\Tr_B\{\rho_{RB}\}=\rho_R$ such that $\rho_R\in \mc{D}(\mc{H}_R)$. The notation $B^n:= B_1B_2\cdots B_n$ denotes a joint system consisting of $n$~subsystems, each of which is isomorphic to Hilbert space $\mc{H}_B$. A purification of a density operator $\rho_B\in \mc{D}\(\mc{H}_B\)$ is a pure state $\psi^\rho_{EB}\in \mc{D}\(\mc{H}_{EB}\)$
such that $\Tr_E\{\psi^\rho_{EB}\}=\rho_B$, where $E$ is called the purifying system. The state $\Phi_{RB}\in \mc{D}\(\mc{H}_{RB}\)$ denotes a bipartite, maximally entangled state, and $\pi_B\in\mc{D}\(\mc{H}_B\)$ denotes the maximally mixed state. 

The evolution of a quantum state is described by a quantum channel. A quantum channel $\mc{N}_{B'\to B}$ is a completely positive, trace-preserving (CPTP) map $\mc{N}:\mc{L}_+(\mc{H}_{B'})\to \mc{L}_+(\mc{H}_B)$. The Choi state $\omega_{RB}$ of a quantum channel $\mc{N}_{B'\to B}$ is defined as 
\begin{equation}
\omega_{RB}:=(\id_{R}\otimes\mc{N}_{B'\to B})(\Phi_{RB'}).
\end{equation} 
A \textit{memory cell} $\{\mc{N}^x\}_{x\in\mc{X}}$ is defined as a set of quantum channels; i.e., $\mc{N}^{x}:\mc{L}_+(\mc{H}_{B'})\to\mc{L}_+(\mc{H}_{B})$ is a quantum channel  for all $x\in\mc{X}$, where $\mc{X}$ is an alphabet. 

Let $U^\mc{N}_{B'\to BE}$ denote an isometric extension of a quantum channel $\mc{N}_{B'\to B}$, which by definition means that
\begin{equation}
\Tr_E\!\left\{U^\mc{N}_{B'\to BE}\rho_{B'}\left(U^\mc{N}_{B'\to BE}\right)^\dagger\right\}=\mathcal{N}_{B'\to B}(\rho_{B'}),  
\end{equation}
for all $\rho_{B'}\in \mc{D}(\mc{H}_{B'})$,
along with the following conditions
for $U_\mc{N}$ to be 
an isometry: 
\begin{equation}
U^\dagger_\mc{N}U_\mc{N}=I_{B'}, \ \text{and} \ U_\mc{N}U^\dagger_\mc{N}=\Pi_{BE},
\end{equation}
where $\Pi_{BE}$ is a projection onto a subspace of the Hilbert space $\mc{H}_{BE}$. A positive operator-valued measure (POVM) is a collection of positive semi-definite operators $\{\Lambda^x\}_{x\in\mc{X}}$ such that $\sum_{x\in\mc{X}}\Lambda^x=I$. 

The cumulative distribution function corresponding to the standard normal random variable is defined as
\begin{equation}
\Phi(a):=\int_{-\infty}^a\frac{1}{\sqrt{2\pi}}\exp\!\(-\frac{1}{2}x^2\)\ dx.
\end{equation}
Its inverse is also useful for us and is defined as $\Phi^{-1}(a):=\sup\left\{a\in\mathbb{R}|\Phi(a)\leq \varepsilon\right\}$, which reduces to the usual inverse for $\varepsilon\in(0,1)$.

Throughout we denote probability distributions of any random variables like $X$ and $Y$ by $p_X(x)$ and $p_Y(y)$, respectively.

 \subsection{Entropies and generalized divergences}

The quantum entropy of a density operator $\rho_{B}$ is defined as 
\begin{equation}
H(B)_\rho:= H(\rho_B)= -\Tr[\rho_B\log_2\rho_B].
\end{equation}
The conditional quantum entropy $H({B'}\vert B)_\rho$ of a density operator $\rho_{{B'}B}$ of a joint system ${B'}B$ is defined as
\begin{equation}
H({B'}\vert B)_\rho := H({B'}B)_\rho-H(B)_\rho.
\end{equation}
The quantum relative entropy of two quantum states is a measure of their distinguishability. For $\rho\in\mc{D}(\mc{H})$ and $\sigma\in\mc{L}_+(\mc{H})$, it is defined as~\cite{Ume62} 
\begin{equation}
D(\rho\V \sigma):= \left\{ 
\begin{tabular}{c c}
$\Tr\{\rho[\log_2\rho-\log_2\sigma]\}$, & $\supp(\rho)\subseteq\supp(\sigma)$\\
$+\infty$, &  otherwise.
\end{tabular} 
\right.
\end{equation}
The quantum relative entropy is non-increasing under the action of positive trace-preserving maps \cite{MR15}, which is the statement that $D(\rho\V\sigma)\geq D(\mc{N}(\rho)\V\mc{N}{(\sigma)})$ for any two density operators $\rho$ and $\sigma$ and a positive trace-preserving map $\mc{N}$ (this inequality applies to quantum channels as well \cite{lindblad75}, since every completely positive map is also a positive map by definition).
The relative entropy variance $V(\rho\Vert\sigma)$ of density operators $\rho$ and $\sigma$ is defined as \cite{TH12,li12}
\begin{equation}
V(\rho\Vert\sigma):=\Tr\{\rho\[\log_2\rho-\log_2\sigma-D(\rho\Vert\sigma)\]^2\}.
\end{equation}

The quantum mutual information $I(R;B)_\rho$ is a measure of correlations between quantum systems $R$ and $B$ in a state $\rho_{RB}$. It is defined as
\begin{align}
I(R;B)_\rho &:=\inf_{\sigma_B\in\mc{D}(\mc{H}_B)}D(\rho_{RB}\Vert\rho_R\otimes\sigma_B)\\
& =H(R)_\rho+H(B)_\rho-H(RB)_\rho.
\end{align}
The quantum conditional mutual information $I(R;B\vert C)_\rho$ of a tripartite density operator $\rho_{RBC}$ is defined as
\begin{align}
I(R;B\vert C)_\rho &:=H(R\vert C)_\rho+H(B\vert C)_\rho-H(RB\vert C)_\rho.
\end{align}
It is known that quantum entropy, quantum mutual information, and conditional quantum mutual information are all non-negative quantities.

A quantity is called a generalized divergence \cite{PV10,SW12} if it satisfies the following monotonicity (data-processing) inequality for all density operators $\rho$ and $\sigma$ and quantum channels $\mc{N}$:
\begin{equation}\label{eq:gen-div-mono}
\mathbf{D}(\rho\Vert \sigma)\geq \mathbf{D}(\mathcal{N}(\rho)\Vert \mc{N}(\sigma)).
\end{equation}
As a direct consequence of the above inequality, any generalized divergence satisfies the following two properties for an isometry $U$ and a state~$\tau$ \cite{WWY14}:
\begin{align}
\mathbf{D}(\rho\Vert \sigma) & = \mathbf{D}(U\rho U^\dag\Vert U \sigma U^\dag),\label{eq:gen-div-unitary}\\
\mathbf{D}(\rho\Vert \sigma) & = \mathbf{D}(\rho \otimes \tau \Vert \sigma \otimes \tau).\label{eq:gen-div-prod}
\end{align}
One can define a mutual-information-like quantity for any quantum state $\rho_{RB}$ as
\begin{equation}
I_{\mathbf{D}}(R;B)_\rho :=\inf_{\sigma_B\in\mc{D}(\mc{H}_B)}\mathbf{D}(\rho_{RB}\Vert \rho_R\otimes\sigma_B).
\end{equation}

The trace distance between two density operators $\rho,\sigma\in\mc{D}(\mc{H})$ is equal to $\Vert \rho-\sigma\Vert_1$, where $\Vert T\Vert_1=\Tr\{\sqrt{T^\dag T}\}$. The fidelity between two states $\rho,\sigma\in\mc{D}(\mc{H})$ is defined as $F(\rho,\sigma):=\Vert \sqrt{\rho}\sqrt{\sigma}\Vert_1^2$ \cite{U76}. The trace distance and the negative logarithm of the fidelity are particular examples of generalized divergences.

The sandwiched R\'enyi relative entropy  \cite{MDSFT13,WWY14} is denoted as $\wt{D}_\alpha(\rho\V\sigma)$  and defined for
$\rho\in\mc{D}(\mc{H})$, $\sigma\in\mc{L}_+(\mc{H})$, and  $\forall \alpha\in (0,1)\cup(1,\infty)$ as
\begin{equation}\label{eq:def_sre}
\wt{D}_\alpha(\rho\V \sigma):= \frac{1}{\alpha-1}\log_2 \Tr\left\{\left(\sigma^{\frac{1-\alpha}{2\alpha}}\rho\sigma^{\frac{1-\alpha}{2\alpha}}\right)^\alpha \right\} ,
\end{equation}
but it is set to $\wt{D}_\alpha(\rho\Vert\sigma)=+\infty$ for $\alpha\in(1,\infty)$ if $\supp(\rho)\nsubseteq \supp(\sigma)$.
The sandwiched R\'enyi relative entropy obeys the following ``monotonicity in $\alpha$'' inequality \cite{MDSFT13}:
\begin{equation}\label{eq:mono_sre}
\wt{D}_\alpha(\rho\V\sigma)\leq \wt{D}_\beta(\rho\V\sigma) \end{equation}
 if  $ \alpha\leq \beta$,  for $ \alpha,\beta\in(0,1)\cup(1,\infty)$.
The following lemma states that the sandwiched R\'enyi relative entropy $\wt{D}_\alpha(\rho\V\sigma)$ is a particular generalized divergence for certain values of $\alpha$. 
\begin{lemma}[\cite{FL13,Bei13}]
Let $\mc{N}_{B'\to B}$ be a quantum channel   and let $\rho_{B'}\in\mc{D}(\mc{H}_{B'})$ and $\sigma_{B'}\in \mc{L}_+(\mc{H}_{B'})$. Then, for all $\alpha\in \[1/2,1\)\cup (1,\infty)$, the following inequality holds
\begin{equation}
\wt{D}_\alpha(\rho\V\sigma)\geq \wt{D}_\alpha(\mc{N}(\rho)\V\mc{N}(\sigma)) .
\end{equation}
\end{lemma}

In the limit $\alpha\to 1$, the sandwiched R\'enyi relative entropy $\wt{D}_\alpha(\rho\V\sigma)$ converges to the quantum relative entropy \cite{MDSFT13,WWY14}
\begin{equation}\label{eq:mono_renyi}
\lim_{\alpha\to 1}\wt{D}_\alpha(\rho\V\sigma):= D_1(\rho\V\sigma)=D(\rho\V\sigma).
\end{equation} 
The sandwiched  R\'enyi mutual information $\wt{I}_\alpha(R;B)_\rho$ is defined as \cite{Bei13,GW13}
\begin{equation}
\wt{I}_\alpha(R;B)_\rho:=\min_{\sigma_B}\wt{D}_\alpha(\rho_{RB}\V\rho_R\otimes\sigma_B).
\end{equation}

Another generalized divergence we make use of is the $\varepsilon$-hypothesis-testing divergence \cite{BD10,WR12},  defined as
\begin{multline}
D^\varepsilon_h\!\(\rho\Vert\sigma\):= \\
-\log_2\inf_{\Lambda}\{\Tr\{\Lambda\sigma\}:\ 0\leq\Lambda\leq I \wedge\Tr\{\Lambda\rho\}\geq 1-\varepsilon\},
\label{eq:hypo-test-div}
\end{multline}
for $\varepsilon\in[0,1]$ and $\rho,\sigma\in\mc{D}(\mc{H})$.

\subsection{Local operations and classical communication (LOCC)}\label{sec:LOCC}

A round of LOCC (or LOCC channel) between two spatially separated parties Alice $A$ and Bob $B$ consists of an arbitrarily large, yet finite number of compositions of the following \cite{BDSW96,CLM+14}:
\begin{enumerate}
\item Alice performs a quantum instrument \cite{W15book} on her system $A$. She forwards the classical output~$x$ to Bob, who then performs a quantum channel on system $B$ conditioned on the classical output $x$. This sequence of actions realizes the following quantum channel:
\begin{equation}
\sum_x\mc{E}^x_A\otimes\mc{F}^x_B,
\end{equation}  
where $\{\mc{E}^x_A\}_x$ is a collection of completely positive, trace non-increasing maps such that $\sum_x\mc{E}^x_A$ is a quantum channel and $\{\mc{F}^x_B\}_x$ is a collection of quantum channels.
\item The situation is reversed, with Bob performing a quantum instrument and forwarding the classical output $y$ to Alice. Alice performs a quantum channel conditioned on the classical output~$y$. This sequence of actions realizes the following quantum channel: 
\begin{equation}
\sum_y\mc{E}^y_B\otimes\mc{F}^y_A. 
\end{equation} 
\end{enumerate}

\subsection{Channels with symmetry}

\label{sec:symmetry}
Consider a finite group $G$. For every $g\in G$, let $g\to U_{B'}(g)$ and $g\to V_B(g)$ be projective unitary representations of $g$ acting on the input space $\mc{H}_{B'}$ and the output space $\mc{H}_B$ of a quantum channel $\mc{N}_{B'\to B}$, respectively. A quantum channel $\mc{N}_{B'\to B}$ is covariant with respect to these representations if the following relation is satisfied \cite{Hol02,Hol06,H13book}:
\begin{equation}\label{eq:cov-condition}
\mc{N}_{B'\to B}\!\(U_{B'}(g)\rho_{B'} U_{B'}^\dagger(g)\)=V_B(g)\mc{N}_{B'\to B}(\rho_{B'}) V_B^\dagger(g), 
\end{equation}
for all $\rho_{B'}\in\mc{D}(\mc{H}_{B'})$ and for all $g\in G$.
For an isometric extension of the above channel $\mc{N}$, there exists a unitary representation $W_E(g)$ acting on the environment Hilbert space $\mc{H}_E$ \cite{Hol06}, such that
for all $g\in G$
\begin{multline}\label{eq:iso-covariant}
\mc{U}^\mc{N}_{B'\to BE}\!\({U_{B'}(g)\rho_{B'}U^\dagger_{B'}(g)}\)=
\\
\(V_B(g)\otimes W_E(g)\)\left[\mc{U}^\mc{N}_{B'\to BE}\(\rho_{B'}\)\right]\(V^\dagger_B(g)\otimes W^\dagger_E(g)\).
\end{multline}
A simple proof of this statement is available in \cite[Appendix~A]{DBW17}. 

In our paper, we define covariant channels in the following way:
\begin{definition}[Covariant channel]\label{def:covariant}
A quantum channel is covariant if it is covariant with respect to a group $G$ which has a representation $U(g)$, for all $g\in G$, on $\mc{H}_{B'}$ that is a unitary one-design; i.e., the map  $\frac{1}{|G|}\sum_{g\in G}U(g)(\cdot)U^\dagger(g)$ always outputs the maximally mixed state for all input states. 
\end{definition}

\begin{definition}[Teleportation-simulable channel \cite{BDSW96,HHH99,Mul12}]\label{def:tel-sim}
A channel $\mc{N}_{B'\to B}$ is teleportation-simulable with associated resource state $\omega_{RB}\in\mc{D}\(\mc{H}_{RB}\)$ if there exists an LOCC channel $\mc{L}_{RB'B\to B}$    such that
\cite[Eq.~(11)]{HHH99}
\begin{equation}
\mc{N}_{B'\to B}\(\rho_{B'}\)=\mc{L}_{RB' B\to B}\(\rho_{B'}\otimes\omega_{RB}\),
\end{equation}
for all $\rho_{B'}\in\mc{D}\(\mc{H}_{B'}\)$,
 and the bipartite cut of the LOCC channel is $RB'|B$. 
 A particular example of an LOCC channel could be  a generalized teleportation protocol \cite{Wer01}.
\end{definition}

The following lemma from \cite[Section~7]{CDP09} extends the developments in \cite{Wer01,Wolf12,LM15}. See also \cite[Appendix~A]{WTB16} for a short proof of the following lemma.

\begin{lemma}[\cite{CDP09}]\label{thm:cov-tel-sim-channel}
All covariant channels (Definition~\ref{def:covariant}) are teleportation-simulable with respect to the resource state $\mathcal{N}_{B'\to B}(\Phi_{RB'})$.
\end{lemma}

%We say that a collection of channels $\{\mc{E}^x_{B'\to B}\}_x$ is environment-parametrized if  
%there exists a set of states $\{\theta^x_E\}_x$, called environment states, and
%$\mc{F}_{B'E\to B}$ a quantum channel, called an interaction channel, such that $\mc{E}^x_{B'\to B}$ can be realized as follows:
%\begin{equation}
%\mc{E}^x_{B'\to B}(L_{B'}):=\mc{F}_{B'E\to B}(L_{B'}\otimes\theta^x_E),
%\end{equation}
%where $L_{B'}$ is an operator acting on $\mc{H}_{B'}$.  

%\begin{remark}
%One can observe from Definition~\ref{def:tel-sim}  that a teleportation-simulable channel is a particular kind of environment-parametrized channel in which
%$\omega_{RB}$ is the environment state and
%$\mc{L}_{R{B'}B\to B}$ is the interaction channel.  
%\end{remark}

\section{Environment-parametrized memory cells}\label{sec:channel-symmetry}

We now introduce a broad class of memory cells that we call environment-parametrized memory cells with associated environment states, and we discuss two classes of memory cells that are particular kinds of environment-parametrized memory cells. Consider an alphabet $\mc{X} := \{ 1, \ldots, |\mc{X}|\}$, where $|\mc{X}|$ is some positive integer. 

\begin{definition}[Environment-parametrized memory cell]\label{def:env-cell}
A memory cell $\mc{E}_{\mc{X}}=\{\mc{E}_{{B'}\to B}^x\}_{x\in\mc{X}}$ is  environment-parametrized with associated environment states $\{\theta^x_{E}\}_{x\in\mc{X}}$ if there exists  a fixed interaction channel
$\mc{F}_{{B'}E\to B}$
such that for all input states $\rho_{B'}$ and
$\forall x\in\mc{X}$
\begin{equation}\label{eq:env-cell}
\mc{E}^x_{{B'}\to B}(\rho_{B'}) =\mc{F}_{{B'}E\to B}(\rho_{B'}\otimes\theta^x_E).
\end{equation}
\end{definition}

This notion is related to the notion of programmable channels, used in the context of quantum computation \cite{DP05}. We should clarify that  \textit{any} memory cell
$\{\mathcal{E}^x_{B' \to B}\}_{x \in \mathcal{X}}$ is environment-parametrized in a trivial way, i.e., with trivial associated environment states. To see this, one can set
\begin{align}
\theta^x_E & = \vert x\rangle\langle x\vert_E, \\
\mc{F}_{{B'}E\to B}(\cdot) & = \sum_{x} \langle x | (\cdot)_E| x\rangle _E \mathcal{E}^x_{B'\to B}(\cdot)_{B'}.
\end{align}
However, one of the main goals of our paper is to establish upper bounds on reading rates of memory cells, and the above construction gives trivial bounds. Thus, when employing the concept of environment-parametrized memory cells, one seeks to find associated environment states $\theta^x_E$ that have the least distinguishability as possible while still being able to realize the memory cell via a common interaction channel $\mc{F}_{{B'}E\to B}$. We provide several non-trivial examples of environment-parametrized memory cells in Section~\ref{sec:example}.

 %Also, one can observe from Definition~\ref{def:tel-sim}  that a teleportation-simulable channel is a particular kind of environment-parametrized channel in which
%$\omega_{RB}$ is the environment state and
%$\mc{L}_{R{B'}B\to B}$ is the interaction channel.  

\begin{definition}[Jointly teleportation-simulable memory cell]\label{def:tel-cell}
A memory cell $\mc{T}_{\mc{X}}=\{\mc{N}^x_{{B'}\to B}\}_{x\in \mc{X}}$  is jointly teleportation-simulable with associated resource states $\{\omega^x_{RB}\}_{x\in\mc{X}}$ if there exists  an LOCC channel $\mc{L}_{{B'} RB\to B}$ such that, for all input states $\rho_{{B'}}$ and
$\forall x\in\mc{X}$
\begin{equation}\label{eq:tel-cell}
 \mc{N}^x_{{B'}\to B}(\rho_{B'})=\mc{L}_{{B'}RB\to B}(\rho_{B'}\otimes\omega^x_{RB}),
\end{equation}  
 where the LOCC channel input is with respect to the bipartition $RB'|B$. 
\end{definition}

\begin{definition}[Jointly covariant memory cell]
A memory cell $\mc{M}_{\mc{X}}=\{\mc{R}^x_{{B'}\to B}\}_{x\in\mc{X}}$  is  jointly covariant if there exists a group $G$ such that for all $x\in\mc{X}$, the channel $\mc{R}^x$ is a covariant channel with respect to the group $G$ (cf., Definition~\ref{def:covariant}).
\end{definition}

\begin{proposition}
Any jointly covariant memory cell $\mc{M}_{\mc{X}}=\{\mc{R}^x_{{B'}\to B}\}_{x\in\mc{X}}$ is jointly teleportation-simulable with respect to the set $\{\mc{R}^x_{{B'}\to B}(\Phi_{R{B'}})\}_{x\in\mc{X}}$ of resource states.
\end{proposition}
\begin{proof}
For a jointly covariant memory cell with respect to a group $G$, all the channels $\mc{R}^x_{{B'}\to B}$ are jointly teleportation-simulable with respect to the resource states $\mc{R}^x_{{B'}\to B}(\Phi_{R{B'}})$ by using a fixed POVM $\{E^g_{B'' R}\}_{g\in G}$, similar to that defined in \cite[Equation (A.4), Appendix A]{WTB16}. See \cite[Appendix A]{WTB16} for an explicit proof.
\end{proof}

\begin{remark}\label{rem:env-tel}
Any jointly teleportation-simulable memory cell with associated resource states $\{\omega^x_{RB}\}_{x\in\mc{X}}$ is environment-parametrized with $\{\omega^x_{RB}\}_{x\in\mc{X}}$ being the associated environment states, an observation that is a direct consequence of definitions. This implies that all jointly covariant memory cells are also environment-parametrized with associated environment states $\{\mc{R}^x_{{B'}\to B}(\Phi_{R{B'}})\}_{x\in\mc{X}}$. 
\end{remark}

\section{Quantum reading protocols and quantum reading capacity}

\label{sec:q-r-p}

In a quantum reading protocol, we consider an encoder and a reader (decoder). An encoder is one who encodes a message onto a physical device that is delivered to Bob, a receiver, whose task it is to read the message. We also refer to Bob as the reader. The quantum reading task comprises the estimation of a message encoded in the form of a sequence of quantum channels chosen from a given memory cell $\{\mc{N}^x_{{B'}\to B}\}_{x\in\mc{X}}$, where $\mc{X}$ is an alphabet. In the most general setting considered in our paper, the reader can use an adaptive strategy for quantum reading.

\begin{figure}[ptb]
\begin{center}
\includegraphics[
width=\linewidth
]{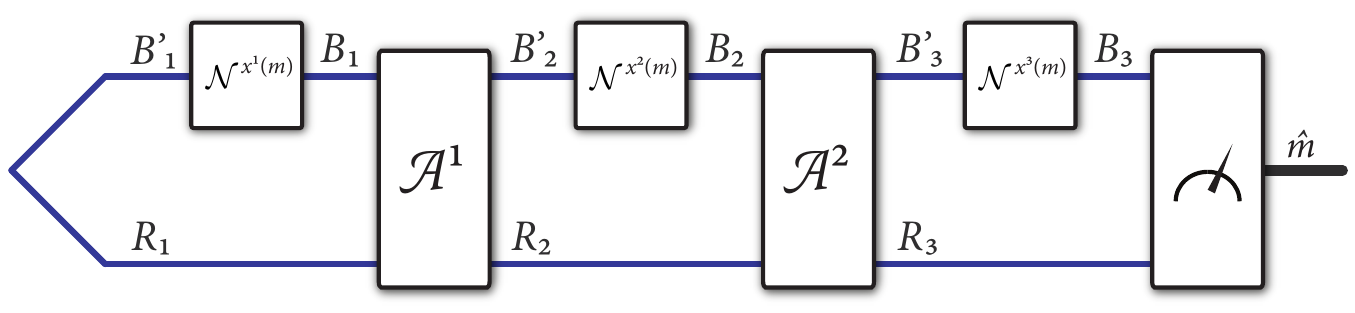}
\end{center}
\caption{The figure depicts a quantum reading protocol that calls a memory cell three times to decode the message $m$ as $\hat{m}$. See the discussion in Section~\ref{sec:q-r-p} for a detailed description of a quantum reading protocol.}%
\label{fig:reading-protocol}%
\end{figure}

Both the encoder and the reader agree upon a memory cell $\mathcal{S}_{\mathcal{X}}=\{\mathcal{N}^x_{B'\to B}\}_{x\in\mathcal{X}}$ before executing the reading protocol. 
We consider a classical message set $\mathcal{M}=\{1,2,\ldots,|\mathcal{M}|\}$, and let $M$ be an associated system denoting a classical register for the message. The encoder encodes a message $m\in\mathcal{M}$ using a codeword
$x^n(m)=(x_1(m),x_2(m),\ldots, x_n(m))$
of length $n$, where $x_i(m)\in\mathcal{X}$ for all $ i\in\{1,2,\ldots,n\}$. Each quantum channel in a codeword, each of which represents one part of the stored information, is only read once. Each codeword identifies with a corresponding sequence of quantum channels chosen from the memory cell $\mathcal{S}_{\mathcal{X}}$:
\begin{equation}
\left(\mathcal{N}^{x_1(m)}_{B'_1\to B_1}, \mathcal{N}^{x_2(m)}_{B'_2\to B_2},\ldots,\mathcal{N}^{x_n(m)}_{B'_n\to B_n}\right).
\end{equation}

An adaptive decoding strategy $\mathcal{J}_{\mc{S}_\mc{X}}$ makes $n$ calls to the memory cell $\mc{S}_{\mc{X}}$. It is specified in terms of a transmitter state $\rho_{R_1B'_1}$, a set of adaptive, interleaved channels $\{\mc{A}^i_{R_iB_i\to R_{i+1}B'_{i+1}}\}_{i=1}^{n-1}$, and a final quantum measurement $\{\Lambda^{\hat{m}}_{R_nB_n}\}_{\hat{m}\in\mc{M}}$ that outputs an estimate $\hat{m}$ of the message $m$. The strategy begins with Bob preparing the input state $\rho_{R_1B'_1}$ and sending the $B'_1$ system into the channel $\mc{N}^{x_1(m)}_{B'_1\to B_1}$. The channel outputs the system $B_1$, which is available to Bob. He adjoins the system $B_1$ to the system $R_1$ and applies the channel $\mc{A}^1_{R_1B_1\to R_2B'_2}$. The channel $\mc{A}^i_{R_iB_i\to R_{i+1}B'_{i+1}}$ is called adaptive because it can take an action conditioned on the information in the system $B_i$, which itself might contain partial information about the message $m$. Then, he sends the system $B'_2$ into the second use of the channel $\mc{N}^{x_2(m)}_{B'_2\to B_2}$, which outputs a system $B_2$. The process of successively using the channels interleaved by the adaptive channels  continues $n-2$ more times, which results in the final output systems $R_n$ and $B_n$ with Bob. Next, he performs a measurement $\{\Lambda^{\hat{m}}_{R_nB_n}\}_{\hat{m}\in\mc{M}}$ on the output state $\rho_{R_n B_n}$, and the measurement outputs an estimate $\hat{m}$ of the original message~$m$.
See Figure~\ref{fig:reading-protocol} for a depiction of a quantum reading protocol.

It is apparent that a non-adaptive strategy is a special case of an adaptive strategy in which the reader does not perform any adaptive channels and instead uses $\rho_{RB^{'n}}$ as the transmitter state with each $B'_i$ system passing through the corresponding channel $\mc{N}^{x_i(m)}_{B'_i\to B_i}$ and $R$ being an idler system. The final step in such a non-adaptive strategy is to perform a decoding measurement on the  joint system $RB^n$. 

As we argued previously,
it is natural to consider the use of an adaptive strategy for a quantum reading protocol because
the channel input and output systems are in the same physical location. In a quantum reading protocol, the reader assumes the role of both the transmitter and receiver. 

\begin{definition}[Quantum reading protocol]
\label{def:QR}
An $(n,R,\varepsilon)$ quantum reading protocol for a memory cell $\mc{S}_{\mathcal{X}}$ is defined by an encoding map $\mc{E}:\mc{M}\to \mc{X}^{\times n}$ and an adaptive strategy $\mathcal{J}_{\mc{S}_\mc{X}}$ with measurement $\{\Lambda_{R_n B_n}^{\hat{m}}\}_{\hat{m}\in\mc{M}}$. The protocol is such that the average success probability is at least $1-\varepsilon$, where $\varepsilon\in(0,1)$:
\begin{multline}
1- \varepsilon \leq 1 - p_{\operatorname{err}} := \\
\frac{1}{|\mc{M}|} 
\sum_{m}\Tr\big\{\Lambda^{(m)}_{R_nB_n}\big(\mc{N}^{x_n(m)}_{B'_n\to B_n}\circ\mc{A}^{{n-1}}_{R_{n-1}B_{n-1}\to R_nB'_n}\circ\\
\cdots\circ\mc{A}^{1}_{R_1B_1\to R_2B'_2}\circ\mc{N}^{x_1(m)}_{B'_1\to B_1}\big)(\rho_{R_1{B'}_1})\big\}.
\end{multline} 
The rate $R$ of a given $(n,R,\varepsilon)$ quantum reading protocol is equal to the number of bits read per channel use:
\begin{equation}
R:=\frac{1}{n}\log_2|\mathcal{M}|.
\end{equation}
\end{definition}

To arrive at a definition of quantum reading capacity, we demand that there exist a sequence of reading protocols, indexed by $n$, for which the error probability $p_e\to 0$ as $n\to \infty$ at a fixed rate~$R$. In more detail, consider the following definitions:

\begin{definition}[Achievable rate]
A rate $R$ is called achievable if for all $ \varepsilon\in (0,1]$, $\delta>0$, and sufficiently large $n$, there exists an $(n,R-\delta,\varepsilon)$ code. 
\end{definition}

\begin{definition}[Quantum reading capacity]\label{def:capacity}
The quantum reading capacity $\mc{C}(\mc{S}_{\mc{X}})$ of a memory cell $\mc{S}_{\mc{X}}$ is defined as the supremum of all achievable rates.
\end{definition}

We also provide the following formal definitions for strong converse rates and the strong converse reading capacity:

\begin{definition}[Strong converse  rate]
A rate $R$ is called a strong converse rate if for all $ \varepsilon\in [0,1)$, $\delta>0$, and sufficiently large $n$, there does not exist an $(n,R+\delta,\varepsilon)$ code. 
\end{definition}

\begin{definition}[Strong converse quantum reading capacity]\label{def:str-conv-capacity}
The strong converse quantum reading capacity $\widetilde{\mc{C}}(\mc{S}_{\mc{X}})$ of a memory cell $\mc{S}_{\mc{X}}$ is defined as the infimum of all strong converse rates.
\end{definition}

The following inequality is a direct consequence of the definitions:
\begin{equation}
\mc{C}(\mc{S}_{\mc{X}}) \leq \widetilde{\mc{C}}(\mc{S}_{\mc{X}}).
\end{equation}

\section{Fundamental limits on quantum reading capacities}

\label{sec:converse bounds}
In this section, we establish second-order  and strong converse bounds for any environment-parametrized memory cell with associated environment states (Definition~\ref{def:env-cell}). We also establish general weak converse (upper) bounds on various reading capacities.  
  
\subsection{Converse bounds for environment-parametrized memory cells}

In this section, we provide upper bounds on the performance of quantum reading of environment-parametrized memory cells with associated environment states.  To begin with,
let us consider an $(n,R,\varepsilon)$ quantum reading protocol of an environment-parametrized memory cell $\mc{E}_\mc{X}=\{\mc{E}^x\}_{x\in\mc{X}}$, as given in Definition~\ref{def:env-cell}.
 The structure of reading protocols involving adaptive channels simplifies immensely for  memory cells that are teleportation simulable and more generally environment-parametrized. 
 This is a consequence of observations made in \cite[Section V]{BDSW96}, \cite[Theorem 14 and Remark 11]{Mul12}, and \cite{DM14}.
 For such memory cells, a quantum reading protocol can be simulated by one in which every channel use is replaced by the encoder preparing the environment state $\theta^{x_i(m)}_{E}$ from \eqref{eq:env-cell} and then interacting the channel input with the interaction channel $\mc{F}_{B'E\to B}$. Critically, each interaction channel
$\mc{F}_{B'E\to B}$ 
 is independent of the message $m\in\mc{M}$. Let 
\begin{equation}
\theta^{x^n(m)}_{E^n}:= \bigotimes_{i=1}^n\theta^{x_i(m)}_{E}
\end{equation}
denote the environment state needed for the simulation of all $n$ of the channel uses in the quantum reading protocol.
%This allows Bob in the transmitter role to take \textquotedblleft full control\textquotedblright of the channel, delaying or advancing its use at will.
This leads to the translation of a general quantum reading protocol to one in which all of the rounds of adaptive channels can be delayed until the very end of the protocol, such that the resulting protocol is a non-adaptive quantum reading protocol.
%\iffalse
\begin{figure}[ptb]
\begin{center}
\includegraphics[
width=\linewidth
]{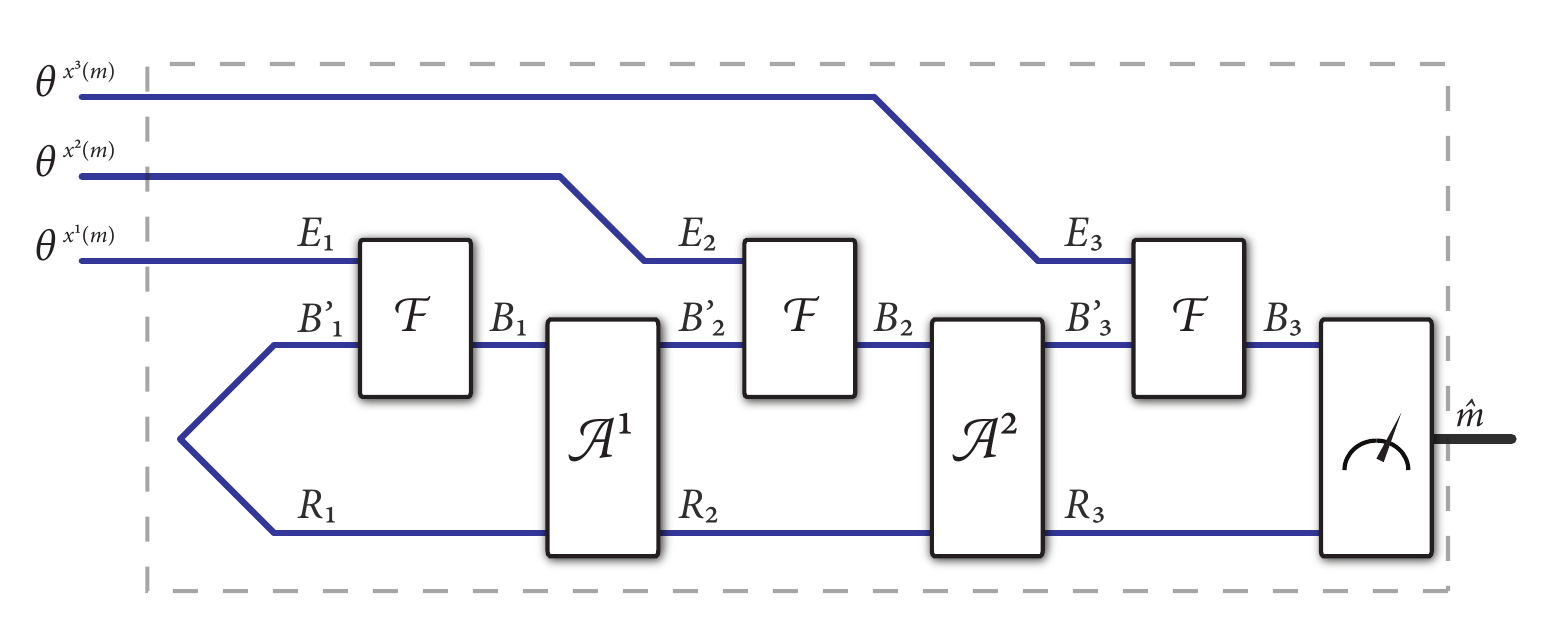}
\end{center}
\caption{The figure depicts how a quantum reading protocol of an environment-parametrized memory cell with associated environment states $\{\theta^x_E\}_{x\in \mathcal{X}}$ can be rewritten as a protocol that tries to decode the message $m$ from the environment states $\theta^{x^n(m)}_{E^n}$. All of the operations inside the dashed lines can be understood as a measurement on the states $\theta^{x^n(m)}_{E^n}$.}%
\label{fig:env-param-reading-protocol}%
\end{figure}
%\fi
The following proposition, holding for any environment-parametrized memory cell with associated environment states, is a direct consequence of observations made in \cite[Section V]{BDSW96}, \cite[Theorem 14 and Remark 11]{Mul12}, and \cite{DM14}. We thus omit a detailed proof, but Figure~\ref{fig:env-param-reading-protocol} clarifies the main idea: any quantum reading protocol of an environment-parametrized memory cell can be rewritten as in Figure~\ref{fig:env-param-reading-protocol}. Inspecting the figure, we see that the protocol can be understood as a non-adaptive decoding of the environment states  $\theta^{x^n(m)}_{E^n}$,
with the decoding measurement constrained to contain the interaction channel $\mc{F}_{B'E\to B}$ interleaved between arbitrary adaptive channels.  Thus, Proposition~\ref{thm:sim-red} establishes that an adaptive strategy used for decoding an environment-parametrized memory cell can be reduced to a particular non-adaptive decoding of the environment states $\theta^{x^n(m)}_{E^n}$. 

\begin{proposition}[Adaptive-to-non-adaptive reduction]\label{thm:sim-red}
Let $\mc{E}_\mc{X}=\{\mc{E}^x_{B'\to B}\}_{x\in\mc{X}}$ be an environment-parametrized memory cell with an associated set of environment states $\{\theta^x_{E}\}_{x\in\mc{X}}$ and a fixed interaction channel $\mc{F}_{B'E\to B}$, as given in Definition~\ref{def:env-cell}.
%Note that each interaction channel $\mc{F}_{AE\to B}$ is independent of $m\in\mc{M}$.
Then any quantum reading protocol as stated in Definition~\ref{def:QR}, which uses an adaptive strategy $\mathcal{J}_{\mc{E}_\mc{X}}$, can be simulated as a non-adaptive quantum reading protocol, in the following sense:
\begin{multline}
\Tr\big\{\Lambda^{\hat{m}}_{E_nB_n}\big(\mc{E}^{x_n(m)}_{B'_n\to B_n}\circ\mc{A}^{{n-1}}_{E_{n-1}B_{n-1}\to E_nB'_n}\circ
\\
\cdots\circ\mc{A}^{1}_{E_1B_1\to E_2B'_2}\circ\mc{E}^{x_1(m)}_{B'_1\to B_1}\big)(\rho_{E_1B'_1})\big\}\\
=\Tr\left\{\Gamma^{\hat{m}}_{E^n}(\bigotimes_{i=1}^n\theta^{x_i(m)}_{E})\right\}, \label{eq:tel-sim-red}
\end{multline}
for some POVM $\{\Gamma^{\hat{m}}_{E^n}\}_{\hat{m}\in \mc{M}}$ that depends on $\mathcal{J}_{\mc{E}_\mc{X}}$.
\end{proposition}

Using the observation in Proposition~\ref{thm:sim-red}, we now show how to arrive at upper bounds on the performance of any reading protocol that uses an environment-parametrized memory cell with associated environment states $\{\theta^x_{E}\}_{x\in\mc{X}}$.

Our proof strategy is to employ a generalized divergence (recall \eqref{eq:gen-div-mono}) to make a comparison between the states involved in the actual  reading protocol and one in which the memory cell is fixed as $\hat{\mc{E}}:=\{\mc{P}_{B'\to B}\}$, containing only a single channel 
with environment state $\hat{\theta}_{E}$
and interaction channel $\mc{F}_{B'E\to B}$.
 The latter reading protocol contains no information about the message $m$.
 Observe that the augmented memory cell $\{\mc{E}_\mc{X},\hat{\mc{E}}\}$ is environment-parametrized with associated environment states $\{\{\theta^x_{E}\}_{x\in\mc{X}},\hat{\theta}_{E}\}$.

One of the main steps that we use in our proof is as follows. Consider the following states:\begin{align}
\sigma_{M\hat{M}} & =\sum_{m\in\mc{M},\atop\hat{m}\in\mc{M}}\frac{1}{|\mc{M}|}|m\>\<m|_M\otimes p_{\hat{M}|M}\(\hat{m}|m\)|\hat{m}\>\<\hat{m}|_{\hat{M}},\label{eq:out-useful}\\
\tau_{M\hat{M}} & =\sum_{m\in\mc{M}}
\frac{1}{|\mc{M}|}
|m\>\<m|_M\otimes \hat{\tau}_{\hat{M}},\label{eq:out-useless}
\end{align}
where we suppose that $p_{\hat{M}|M}(\hat{m}|m)$ is a distribution that results after the final decoding step of an $(n,R,\varepsilon)$ quantum reading protocol, while  $\hat{\tau}_{\hat{M}}$
is a fixed state. By applying
the comparator test $\{\Pi_{M\hat{M}},I_{M\hat{M}}-\Pi_{M\hat{M}}\}$, defined by
\begin{equation}
\Pi_{M\hat{M}} := \sum_{m}|m\>\<m|_M\otimes|m\>\<m|_{\hat{M}},\label{eq:comparator-test}
\end{equation}
and using definitions, we arrive at the following inequalities that hold for an arbitrary
$(n,R,\varepsilon)$ quantum reading protocol:
\begin{align}
\Tr\{\Pi_{M\hat{M}}\sigma_{M\hat{M}}\} & \geq 1-\varepsilon,\\
 \Tr\{\Pi_{M\hat{M}}\tau_{M\hat{M}}\} & = \frac{1}{|\mc{M}|}.
\end{align}
Then by applying the definition of the $\varepsilon$-hypothesis-testing divergence (recall \eqref{eq:hypo-test-div}), we arrive at the following bound, which is a critical first step for us to establish second-order  and strong converse bounds:
\begin{equation}\label{eq:eps-div-M}
D^\varepsilon_h\!\(\sigma_{M\hat{M}}\Vert\tau_{M\hat{M}}\)\geq \log_2|\mc{M}|.
\end{equation}

In the converse proof that follows, the main idea for arriving at an upper bound on performance is to make a comparison between the case in which the message $m$ is encoded in a sequence of quantum channels and the case in which it is not.

\subsubsection{Second-order asymptotics and strong converse}

In this section, we derive second-order asymptotics and strong converse bounds for environment-parametrized memory cells with associated environment states. We begin by deriving a relation between the quantum reading rate and the hypothesis testing divergence.  

\begin{lemma}\label{thm:hyp-test-bound}
The following bound holds for an $(n,R,\varepsilon)$ reading protocol that uses an environment-parametrized memory cell with associated environment states $\{\theta^x_{E}\}_{x\in\mc{X}}$, as stated in Definition~\ref{def:env-cell}:
\begin{equation}
\log_2 |\mc{M}| = nR \leq 
\sup_{p_{X^n}}\inf_{\hat{\theta}} D_h^{\varepsilon}(\theta_{X^n E^n} \Vert \hat{\theta}_{X^n E^n} ) ,
\end{equation}
where
\begin{align}
\theta_{X^n E^n} & := \sum_{x^n\in\mc{X}^n}p_{X^n}(x^n)|x^n\>\<x^n|_{X^n}\otimes \theta^{x^n}_{E^n}, \\
\hat{\theta}_{X^n E^n} & := \sum_{x^n\in\mc{X}^n}p_{X^n}(x^n)|x^n\>\<x^n|_{X^n}\otimes \hat{\theta}_{E}^{\otimes n}.
\end{align}
and  $x^n:= x_1x_2\cdots x_n$ and $\theta_{E^n}^{x^n}=\bigotimes_{i=1}^n\theta_{E}^{x_i}$.
\end{lemma}

\begin{proof}
Our proof begins by applying the observation from Proposition~\ref{thm:sim-red}, which allows us to reduce any adaptive protocol to a non-adaptive one. If the encoder chooses the message $m$ uniformly at random and places it in a system $M$, the output state in \eqref{eq:out-useful} after Bob's decoding measurement in the actual protocol is
\begin{equation}
\sigma_{M\hat{M}} = \sum_{m,\hat{m}}\frac{1}{|\mathcal{M}|}|m\>\<m|_M\otimes \Tr\!\left\{\Gamma_{E^n}^{\hat{m}}\theta^{x^n(m)}_{E^n}\right\}|\hat{m}\>\<\hat{m}|_{\hat{M}},
\end{equation}
where
\begin{equation}
\theta^{x^n(m)}_{E^n} := 
\bigotimes_{i=1}^n\theta^{x_i(m)}_{E}.
\end{equation}
The success probability $p_{\operatorname{succ}}:=1-p_{\operatorname{err}}$ is defined as
\begin{equation}
p_{\operatorname{succ}}:=\frac{1}{|\mc{M}|}\sum_{m\in\mc{M}}\Tr\!\left\{\Gamma_{E^n}^{m}\theta^{x^n(m)}_{E^n}\right\}.
\end{equation} 
The output state in \eqref{eq:out-useless} after Bob's decoding measurement in a reading protocol that uses the memory cell $\hat{\mc{E}}$ is 
\begin{equation}
\tau_{M\hat{M}} = \sum_m\frac{1}{|\mathcal{M}|}|m\>\<m|_M\otimes\sum_{\hat{m}}\Tr\!\left\{\Gamma_{E^n}^{\hat{m}}\hat{\theta}^{\otimes n}_{E}\right\}|\hat{m}\>\<\hat{m}|_{\hat{M}}.
\end{equation}
Then we can proceed with bounding a generalized divergence as follows:
\begin{align}
& \mathbf{D}\!\(
\{p_{\operatorname{succ}}, 
1-p_{\operatorname{succ}}\} \middle \Vert 
\left\{ 1/|\mc{M}|, 1-1/|\mc{M}|\right\}
\) \nonumber\\
&\leq \mathbf{D}\!\(\sigma_{M\hat{M}} \Vert \tau_{M\hat{M}}\) \\
& \leq  \mathbf{D}( \theta_{ME^n }\Vert \hat{\theta}_{ME^n }),
\end{align}
where
\begin{align}
\theta_{ME^n } &  := \sum_{m}\frac{1}{|\mathcal{M}|}|m\>\<m|_M\otimes \theta^{x^n(m)}_{E^n},\\
\hat{\theta}_{ME^n } & := \sum_m\frac{1}{|\mathcal{M}|}|m\>\<m|_M\otimes\hat{\theta}^{\otimes n}_{E}.
\end{align}
The first inequality follows from applying the comparator test in \eqref{eq:comparator-test} to
$\sigma_{M\hat{M}}$ and $\tau_{M\hat{M}}$.
The second inequality follows from the data-processing inequality in \eqref{eq:gen-div-mono} as the final measurement is a quantum channel.
Since the above bound holds for all $\hat{\theta}_{E}$, we conclude that
\begin{multline}
\mathbf{D}\!\(
\{p_{\operatorname{succ}}, 
1-p_{\operatorname{succ}}\} \middle \Vert 
\left\{ 1/|\mc{M}|, 1-1/|\mc{M}|\right\}
\) \leq \\
\inf_{\hat{\theta}}\mathbf{D}( \theta_{ME^n }\Vert \hat{\theta}_{ME^n }) .
\end{multline}
Now optimizing over all input distributions, we arrive at the following general bound:
\begin{multline}
\mathbf{D}\!\(
\{p_{\operatorname{succ}}, 
1-p_{\operatorname{succ}}\} \middle \Vert 
\left\{ 1/|\mc{M}|, 1-1/|\mc{M}|\right\}
\)\leq \\
\sup_{p_{X^n}}\inf_{\hat{\theta}} \mathbf{D}(\theta_{X^n E^n} \Vert \hat{\theta}_{X^n E^n}) \label{eq:gen-div-bound-JP-cells} .
\end{multline}
Observe that the lower bound contains the relevant performance parameters such as success probability and number of messages, while the upper bound is an information quantity, depending exclusively on the memory cell
$\mc{E}_\mc{X}$. 

Substituting the hypothesis testing divergence in the above and applying \eqref{eq:eps-div-M}, we find the following bound for an $(n,R,\varepsilon)$ reading protocol that uses an environment-parametrized memory cell with associated environment states $\{\theta^x_{E}\}_{x\in\mc{X}}$:
\begin{equation}
\log_2 |\mc{M}| = nR \leq 
\sup_{p_{X^n}}\inf_{\hat{\theta}} D_h^{\varepsilon}(\theta_{X^n E^n} \Vert \hat{\theta}_{X^n E^n}).
\end{equation}
This concludes the proof.
\end{proof}

A direct consequence of Lemma~\ref{thm:hyp-test-bound} and \cite[Theorem 4]{TT13} is the following proposition:
\begin{proposition}\label{thm:env-cell-sec-order}
For an $(n,R,\varepsilon)$ quantum reading protocol for an environment-parametrized memory cell $\mc{E}_{\mc{X}}=\{\mc{E}^x\}_{x\in\mc{X}}$ with associated environment states $\{\theta^x_{E}\}_{x\in\mc{X}}$ (as stated in Definition~\ref{def:env-cell}), the following inequality holds 
\begin{equation}
R\leq  \max_{p_X}I(X;E)_\theta +\sqrt{\frac{V_\varepsilon(\mc{E}_{\mc{X}})}{n}}\Phi^{-1}(\varepsilon)+O\!\(\frac{\log n}{n}\),
\end{equation}
where
\begin{equation}\label{eq:cq-theta}
\theta_{XE}=\sum_{x\in\mc{X}}p_X(x)|x\>\<x|_X\otimes\theta^x_E,
\end{equation}
and
\begin{equation}
V_\varepsilon(\mc{E}_{\mc{X}})= \left\{ 
\begin{tabular}{c c}
$\min_{p_X\in P({\mc{E}})} V(\theta_{XE}\V\theta_X\otimes\theta_{E} )$,  & $\varepsilon \in  (0,1/2]$\\
$\max_{p_X\in P({\mc{E}})} V(\theta_{XE}\V\theta_X\otimes\theta_{E} )$, & $\varepsilon \in (1/2,1)$
\end{tabular}
\right.
\end{equation}
where $P({\mc{E}})$ denotes a set $\{p_X\}$ of probability distributions that achieve the maximum in 
$\max_{p_X}I(X;E)_\theta$.
\end{proposition}

\begin{proposition}\label{thm:env-cell-strong}
The success probability $p_{\operatorname{succ}}$ of any $(n,R,\varepsilon)$ quantum reading protocol for an environment-parametrized memory cell $\mc{E}_\mc{X}$ with associated environment states $\{\theta^x_{E}\}_{x\in\mc{X}}$ is bounded
from above
 as 
\begin{equation}
p_{\operatorname{succ}}\leq 2^{-n\sup_{\alpha>1}\(1-\frac{1}{\alpha}\) \(R-\wt{I}_\alpha(\mc{E}_\mc{X})\)},
\end{equation}
where
\begin{equation}
\wt{I}_\alpha(\mc{E}_\mc{X})=\max_{p_X}\wt{I}_\alpha(X;E)_\theta,
\end{equation}
for $\theta_{XE}$ as defined in \eqref{eq:cq-theta}.
\end{proposition}
\begin{proof}
A proof follows by combining the bound in \eqref{eq:gen-div-bound-JP-cells} with the main result of \cite{WWY14} (see also
\cite{DW15}
for arguments about extending the range of $\alpha$ from $(1,2]$ to $(1,\infty)$). 
\end{proof}

\begin{theorem}\label{thm:env-cell-qrc}
The strong converse quantum reading capacity of any environment-parametrized memory cell $\mc{E}_{\mc{X}}=\{\mc{N}^x_{B'\to B}\}_{x\in\mc{X}}$ with associated environment states $\{\theta^x_{E}\}_{x\in\mc{X}}$, as given in Definition~\ref{def:env-cell}, is bounded from above as
\begin{equation}
\widetilde{\mc{C}}(\mc{E}_{\mc{X}})\leq\max_{p_X}I(X;E)_\theta,
\end{equation} 
where  $\theta_{XE}$ is defined in \eqref{eq:cq-theta}.
\end{theorem}
\begin{proof}
The statement follows from Proposition~\ref{thm:env-cell-sec-order}, by taking the limit $n\to \infty$. Alternatively, 
the statement can also be concluded from Definition~\ref{def:capacity} and Proposition~\ref{thm:env-cell-strong}, by taking the limit $\alpha \to 1$.
\end{proof}

\bigskip
Direct consequences of the above theorems and Remark~\ref{rem:env-tel} are the following corollaries: 
\begin{corollary}\label{thm:rate-2nd}
For any $(n,R,\varepsilon)$ quantum reading protocol and jointly teleportation-simulable memory cell $\mc{T}_{\mc{X}}$ with
associated  resource states $\{\omega^x_{RB}\}_{x\in\mc{X}}$ as stated in Definition~\ref{def:tel-sim}, the reading rate $R$ is bounded from above as
\begin{equation}
R\leq\max_{p_X}I(X;RB)_{\omega}+\sqrt{\frac{V_\varepsilon(\mc{T}_{\mc{X}})}{n}}\Phi^{-1}(\varepsilon)+\mc{O}\(\frac{\log n}{n}\),
\end{equation}
where
\begin{equation}
\omega_{XRB}=\sum_{x\in\mc{X}}p_X(x)|x\>\<x|_X\otimes\omega^x_{RB}\label{eq:cq-resource-state}
\end{equation}
and
\begin{multline}
V_\varepsilon(\mc{T}_{\mc{X}})= \\
\left\{ 
\begin{tabular}{c c}
$\min_{p_X\in P({\mc{T}})} V\(\omega_{XRB}\V\omega_X\otimes\omega_{RB} \)$, & $\varepsilon \in  (0,1/2]$\\
$\max_{p_X\in P({\mc{T}})} V\(\omega_{XRB}\V\omega_X\otimes\omega_{RB} \)$,& $\varepsilon \in (1/2,1)$
\end{tabular}
\right.
\end{multline}
In the above, $P({\mc{T}})$ denotes a set $\{p_X\}$ of probability distributions  that are optimal for $\max_{p_X}I(X;RB)_{\omega}$.
\end{corollary}

\begin{corollary}
The strong converse quantum reading capacity of any jointly teleportation-simulable memory cell $\mc{T}_{\mc{X}}=\{\mc{N}^x_{B'\to B}\}_{x\in\mc{X}}$ associated with a set $\{\omega^x_{RB}\}$ of resource states is bounded from above as
\begin{equation}
\widetilde{\mc{C}}(\mc{T}_{\mc{X}})\leq\max_{p_X}I(X;RB)_\omega,
\end{equation} 
where
\begin{equation}\label{eq:omega-XRB}
\omega_{XRB}=\sum_{x\in\mc{X}}p_X(x)|x\>\<x|_X\otimes\omega^x_{RB}.
\end{equation}
\end{corollary}

The capacity bounds given above are tight for a wide variety of channels, as clarified in the following remark:

\begin{remark}\label{rem:cov-cell-achievability}
The quantum reading capacity is achieved for a jointly teleportation-simulable memory cell 
$\mc{T}_{\mc{X}}$
when, for all $x\in\mc{X}$, $\omega^x_{RB}$ is equal to the Choi state of the channel $\mc{N}^x_{B'\to B}$. More precisely, the upper bound in Corollary~\ref{thm:rate-2nd} is achieved in such a case by invoking \cite[Theorem 4]{TT13}.
\end{remark}

\subsection{Weak converse bound for a non-adaptive reading protocol}

In this section, we establish a general weak converse when the strategy employed is non-adaptive.
Consider a state $\rho_{MRB'^{n}}$  of the form
\begin{equation}
\rho_{MRB'^n}=\frac{1}{|\mc{M}|}\sum_{m}|m\>\<m|_M\otimes\rho_{RB'^n}.
\end{equation}
Suppose that $\rho_{RB'^n}$ is purified by the pure state
$\psi_{RSB'^n}$.
Bob passes the transmitter state $\rho_{RB'^n}$ through a codeword sequence $\mathcal{N}^{x^n(m)}_{B'^n\to B^n}:= \bigotimes_{i=1}^n\mc{N}^{x_i(m)}_{B'_i\to B_i}$, where the choice $m$ depends on the classical value~$m$ in the register $M$. Let
\begin{equation}
\mathcal{U}^{\mc{N}^{x^n(m)}}_{B'^n\to B^n E^n} := \bigotimes_{i=1}^n\mathcal{U}^{\mc{N}^{x_i(m)}}_{B'_i\to B_i E_i} ,
\end{equation}
where $\mathcal{U}^{\mc{N}^{x_i(m)}}_{B'_i\to B_i E_i}$ denotes an isometric quantum channel extending ${\mc{N}}^{x_i(m)}_{B'_i\to B_i}$, for all $i\in \{1,2,\ldots n\}$. After the isometric channel acts, the overall state is as follows:
\begin{equation}
\sigma _{MRSB^nE^n}=
\frac{1}{|\mc{M}|}\sum_{m}|m\>\<m|_M\otimes
\mathcal{U}^{\mc{N}^{x^n(m)}}_{B'^n\to B^n E^n}\left(\psi_{RSB'^n}\right).
\end{equation}
Let $\sigma^\prime_{M\hat{M}}=\mc{D}_{RB^n\to \hat{M}}\(\sigma_{MRB^n}\)$ be the output state at the end of protocol after the decoding measurement $\mc{D}$ is performed by Bob. 
Let $\overline{\Phi}_{M\hat{M}}$ denote the maximally classically correlated state:
\begin{equation}\label{eq:cor-state}
\overline{\Phi}_{M\hat{M}}:=\frac{1}{|\mc{M}|}\sum_{m\in\mathcal{M}}|m\>\<m|_M\otimes|m\>\<m|_{\hat{M}}.
\end{equation}

\begin{proposition}
\label{thm:weak-converse-non-adaptive}
The non-adaptive reading capacity of any quantum memory cell $\mc{S}_{\mc{X}}=\{\mc{N}^x_{B'\to B}\}_{\mc{X}}$ is bounded from above as
\begin{equation}
\mc{C}_{\textnormal{non-adaptive}}(\mc{S}_{\mc{X}})\leq \sup_{p_X,\phi_{RB'}}I(XR;B)_\tau,
\end{equation}
where
\begin{align}
\tau_{XRB}&=\sum_{x}p_X(x)|x\>\<x|_X\otimes\mc{N}^x_{B'\to B}(\phi_{R B'}),
\end{align}
and it suffices for $\phi_{R B'}$ to be a pure state such that system $R$ is isomorphic to system $B'$. 
\end{proposition}
\begin{proof}
For any $(n,R,\varepsilon)$ quantum reading protocol using a non-adaptive strategy, we have that
\begin{equation}
\frac{1}{2}\left\Vert \overline{\Phi}_{M\hat{M}}-\sigma^\prime_{M\hat{M}}\right\Vert_1\leq \varepsilon. 
\end{equation}
Then consider the following chain of inequalities:
\begin{align}
& \log_2|\mc{M}|\notag \\
&=I(M;\hat{M})_{\overline{\Phi}}\\
& \leq I(M;\hat{M})_{\sigma^\prime}+f(n,\varepsilon)\\
& \leq I(M;RSB^n)_\sigma + f(n,\varepsilon)\\
& = I(M;RS)_\sigma +I(M;B^n|RS)_\sigma + f(n,\varepsilon)\\
&=I(M;B^n|RS)_\sigma + f(n,\varepsilon)\\
&=H(B^n|RS)_\sigma-H(B^n|RSM)_\sigma + f(n,\varepsilon)\\
&=H(B^n|RS)_\sigma+H(B^n|E^nM)_\sigma + f(n,\varepsilon) \label{eq:converse-1st-block-last-step}
\end{align}
The first inequality follows from the uniform continuity of conditional entropy \cite{AF04,Win16}, where $f(n,\varepsilon)$ is a function of $n$ and the error probability $\varepsilon$ such that $\lim_{\varepsilon\to 0}\lim_{n\to \infty}\frac{f(n,\varepsilon)}{n}=0$ \cite{W15book}. The second inequality follows from data processing. The second equality follows from the chain rule for the mutual information. The third equality follows because the reduced state of systems $M$ and $RS$ is a product state. The fifth equality follows from the duality of the conditional entropy. 
Continuing, we have that
\begin{align}
& \text{Eq. } \eqref{eq:converse-1st-block-last-step} \notag \\
&\leq \sum_{i=1}^n \left[H(B_i|RS)_\sigma +H(B_i|E_iM)_\sigma\right]+ f(n,\varepsilon)\\
&= \sum_{i=1}^n\left[ H(B_i|RS)_\sigma -H(B_i|RSB'_{[n]\setminus \{i\}}M)_\sigma \right]+ f(n,\varepsilon)\\
&=\sum_{i=1}^nI(MB'_{[n]\setminus \{i\}};B_i|RS)_\sigma + f(n,\varepsilon)\\
&\leq \sum_{i=1}^nI(MB'_{[n]\setminus \{i\}}RS;B_i)_\sigma+ f(n,\varepsilon)\\
&= nI(MR^\prime;B|Q)_\sigma+ f(n,\varepsilon)\\
&\leq n\sup_{p_X,\phi_{\wt{R} B'}}I(X\wt{R};B)_\tau+ f(n,\varepsilon).
\end{align}
The first inequality follows from strong subadditivity of quantum entropy \cite{LR73,PhysRevLett.30.434}. The final inequality follows because the average can never exceed the maximum.
In the above, $B'_{[n]\setminus \{i\} }$ denotes the joint system $B'_1{B'}_2\cdots B'_{i-1}{B'}_{i+1}\cdots B'_n$, such that  system $B'_i$ is excluded. Furthermore,
\begin{multline}
\sigma_{MQR^\prime B} = \\
\frac{1}{|\mc{M}|}\frac{1}{n}\sum_{m=1}^{|\mc{M}|}\sum_{i=1}^n|m\>\<m|_M\otimes|i\>\<i|_Q\otimes\mc{N}^{x_i(m)}_{B'_i\to B_i}(\sigma_{R S B'_i B'_{[n]\setminus{i}}}) ,
\end{multline}
where we have introduced an auxiliary classical register $Q$, and  $R^\prime := RSB'_{[n]\setminus i}$. Also,
\begin{equation}\label{eq:cq-phi}
\tau_{X\wt{R} B}=\sum_{x}p_X(x)|x\>\<x|_X\otimes\mc{N}^x_{B'\to B}(\phi_{\wt{R} B'}).
\end{equation}

Now we argue that it is sufficient to take $\phi_{\wt{R}B'}$ to be a pure state. Suppose that $\phi_{\hat{R}B'}$ is a mixed state and let $R^{\prime\prime}$ be a purifying system for it. Then by the data-processing inequality, we have that
\begin{equation}
I(X\hat{R};B)_\tau \leq I(X\hat{R}R^{\prime\prime};B)_\tau,
\end{equation}
where $\tau_{X\hat{R}R''B}$ is a state of the form in \eqref{eq:cq-phi}.
The statement in the theorem about the  reference system  follows from the Schmidt decomposition and the fact that the reference system purifies the system $B'$ being input to the channel.
\end{proof}

\subsection{Weak converse bound for a quantum reading protocol}

In this section, we establish a general weak converse bound for the quantum reading capacity of an arbitrary memory cell.

\begin{theorem}\label{thm:q-r-capacity}
The quantum reading capacity of a quantum memory cell $\mc{S}_{\mc{X}}=\{\mc{N}^x_{B' \to B}\}_{\mc{X}}$ is  bounded from above as
\begin{equation}
\mathcal{C}(\mc{S}_{\mc{X}})\leq \sup_{\rho_{XR{B'}}}\left[I(X;B|R)_\omega-I(X;{B'}|R)_\rho\right],
\label{eq:reading-cap-up-bnd}
\end{equation}
where
\begin{align}
\omega_{XRB}&=\sum_{x}p_X(x)|x\>\<x|_X\otimes\mc{N}^x_{{B'}\to B}(\rho^x_{R{B'}}),\\
\rho_{XR{B'}}&=\sum_{x}p_X(x)|x\>\<x|_X\otimes\rho^x_{R{B'}},
\end{align}
and the dimension of the Hilbert space $\mc{H}_R$ can be unbounded. 
\end{theorem}

\begin{remark}
We should clarify that the upper bound
in \eqref{eq:reading-cap-up-bnd}
is non-negative. A particular choice of
the input state $\rho_{XR{B'}}$ is
\begin{equation}
\rho_{XR{B'}} =
\sum_{x}p_X(x)|x\>\<x|_X\otimes\rho_{R{B'}}.
\end{equation}
Then in this case,
\begin{align}
I(X;B|R)_\omega-I(X;{B'}|R)_\rho & = I(X;RB)_\omega-I(X;R{B'})_\rho \notag \\
 & = I(X;RB)_\omega \geq 0,
\end{align}
with $\omega_{XRB}=\sum_{x}p_X(x)|x\>\<x|_X\otimes\mc{N}^x_{{B'}\to B}(\rho_{R{B'}})$. Thus, we can conclude that
\begin{equation}
\sup_{\rho_{XR{B'}}}\left[I(X;B|R)_\omega-I(X;{B'}|R)_\rho\right] \geq 0.
\end{equation}
\end{remark}

\begin{proof}[Proof of Theorem~\ref{thm:q-r-capacity}]
For any $(n,R,\varepsilon)$ quantum reading protocol as stated in Definition~\ref{def:QR}, we have
\begin{equation}
\frac{1}{2}\left\Vert  \overline{\Phi}_{M\hat{M}}-\sigma^\prime_{M\hat{M}}\right\Vert_1\leq\varepsilon,
\end{equation}
where $\overline{\Phi}_{M\hat{M}}$ is a maximally classically correlated state as in \eqref{eq:cor-state} and 
\begin{equation}
\sigma^\prime_{M\hat{M}}=\mc{D}_{R_nB_n\to\hat{M}}\(\sigma^n_{MR_nB_n}\)
\end{equation}
is the output state at the end of the protocol after Bob performs the final decoding measurement. 
We denote the input state before the $i$th call of the channel  as 
\begin{multline}
\rho^{i}_{M R_i B'_i}
=\frac{1}{|\mc{M}|}\sum_{m\in\mc{M}}|m\>\<m|_M\otimes\mc{A}^{(i-1)}_{R_{i-1}B_{i-1}\to R_iB'_i}\circ\\
\mc{N}^{x_{i-1}(m)}_{B'_{i-1}\to B_{i-1}}\circ\cdots
\cdots\circ\mc{N}^{x_2(m)}_{B'_2\to B_2}\circ \\
\mc{A}^{(1)}_{R_{1}B_{1}\to R_2B'_2}\circ\mc{N}^{x_{1}(m)}_{B'_{1}\to B_{1}}(\rho_{R_1B'_1})\label{eq:state-i-pre-use} ,
\end{multline}
and we denote the output state after the $i$th call of the channel  as
\begin{multline}
\omega^{i}_{M R_i B_i}
=\frac{1}{|\mc{M}|}\sum_{m\in\mc{M}}|m\>\<m|_M\otimes\mc{N}^{x_i(m)}_{B'_i\to B_i}\circ\\
\mc{A}^{(i-1)}_{R_{i-1}B_{i-1}\to R_iB'_i}\circ
\mc{N}^{x_{i-1}(m)}_{B'_{i-1}\to B_{i-1}}\circ\cdots
\cdots\circ\mc{N}^{x_2(m)}_{B'_2\to B_2}\circ\\
\mc{A}^{(1)}_{R_{1}B_{1}\to R_2B'_2}\circ\mc{N}^{x_{1}(m)}_{B'_{1}\to B_{1}}(\rho_{R_1B'_1}).\label{eq:state-i-use}
\end{multline}
 
The initial part of our proof follows steps similar to those in the proof of Proposition~\ref{thm:weak-converse-non-adaptive}.
\begin{align}
& \log_2|\mc{M}| \notag \\
&= I(M;\hat{M})_{\overline{\Phi}}\\
& \leq I(M;\hat{M})_{\sigma^\prime}+f(n,\varepsilon)\\
& \leq I(M;R_n B_n)_{\sigma^n}+f(n,\varepsilon)\\
& = I(M;R_nB_n)_{\omega^n}-I(M;R_1B'_1)_{\rho^1}+f(n,\varepsilon)\\
&= I(M;R_nB_n)_{\omega^n}-I(M;R_nB'_n)_{\rho^n}
\notag \\
& \qquad + I(M;R_nB'_n)_{\rho^n} -I(M;R_{n-1}B'_{n-1})_{\rho^{n-1}}\nonumber\\ &\qquad+I(M;R_{n-1}B'_{n-1})_{\rho^{n-1}}-\cdots-I(M;R_2B'_2)_{\rho^2}\nonumber\\
& \qquad+I(M;R_2B'_2)_{\rho^2}-I(M;R_1B'_1)_{\rho^1}+f(n,\varepsilon)\\
&\leq I(M;R_nB_n)_{\omega^n}-I(M;R_nB'_n)_{\rho^n} \notag \\& \qquad +I(M;R_{n-1}B_{n-1})_{\omega^{n-1}} -I(M;R_{n-1}B'_{n-1})_{\rho^{n-1}}\nonumber\\ &\qquad+I(M;R_{n-2}B_{n-2})_{\omega^{n-2}}-\cdots-I(M;R_2B'_2)_{\rho^2}\nonumber \\
& \qquad +I(M;R_1B_1)_{\omega^1}-I(M;R_1B'_1)_{\rho^1}+f(n,\varepsilon) \label{eq:first-block-wk-conv-adaptive}
\end{align}
The second equality follows because the state $\rho^1$ is product between systems $M$ and $R_1B'_1$. The third equality follows by  adding and subtracting  equal information quantities.   
The third inequality follows from the data-processing inequality: mutual information is non-increasing under the local action of quantum channels. Continuing, we have that
\begin{align}
& \text{Eq. } \eqref{eq:first-block-wk-conv-adaptive} \notag \\
&= \sum_{i=1}^n\left[I(M;R_iB_i)_{\omega^i}-I(M;R_iB'_i)_{\rho^i}\right]+f(n,\varepsilon)\\
& = \sum_{i=1}^n\left[I(M;B_i|R_i)_{\omega^i}-I(M;B'_i|R_i)_{\rho^i}\right]+f(n,\varepsilon)\\
& = n \left[I(M;B|RQ)_{\overline{\omega}}-I(M;{B'}|RQ)_{\overline{\rho}}\right]+f(n,\varepsilon)\\
&\leq n\sup_{\rho_{XR{B'}}}\left[I(X;B|R)_\omega-I(X;{B'}|R)_\rho\right]+f(n,\varepsilon),
\end{align}
The second equality follows from the chain rule for conditional mutual information. The third equality follows by defining the following states:
\begin{align}
\overline{\omega}_{QMRB} & =
\sum_{i=1}^n \frac{1}{n} | i\> \< i|_Q \otimes 
\omega^{i}_{M R_i B_i}, \\ 
\overline{\rho}_{QMR{B'}} & =
\sum_{i=1}^n \frac{1}{n} | i\> \< i|_Q \otimes 
\rho^{i}_{M R_i B'_i}.
\end{align}
%\begin{equation}
%R=\lim_{n\to\infty}\frac{\log|\mc{M}|}{n}\leq \sup_{\{p_X(x),\rho^x_{RA}\}}I(X;B|R)_\rho.
%\end{equation}
The final inequality follows by defining the following states:
\begin{align}
\omega_{XRB}&=\sum_{x}p_X(x)|x\>\<x|_X\otimes\mc{N}^x_{{B'}\to B}(\rho^x_{R{B'}}),\\
\rho_{XR{B'}}&=\sum_{x}p_X(x)|x\>\<x|_X\otimes\rho^x_{R{B'}},
\end{align}
and realizing that the states $\overline{\omega}_{QMRB}$ and $\overline{\rho}_{QMR{B'}}$ are particular examples of the states $\omega_{XRB}$ and $\rho_{XR{B'}}$, respectively, with the identifications $M \to X$ and $QR \to R$.
Putting everything together, we find that
\begin{multline}
\frac{1}{n}\log_2|\mc{M}| \leq 
\sup_{\rho_{XR{B'}}}\left[I(X;B|R)_\omega-I(X;{B'}|R)_\rho\right]\\
+\frac{1}{n} f(n,\varepsilon)
\end{multline}
Taking the limit as $n \to \infty$ and then as $\varepsilon \to 0$ concludes the proof.
\end{proof}

\bigskip
Now we develop a general upper bound on the energy-constrained quantum reading capacity of a beamsplitter memory cell $\mc{B}_{\mc{X}}=\{\mc{B}^{x}\}_{x\in\mc{X}}$, where $x\in\mc{X}$ represents the transmissivity $\eta$ and phase $\phi$ of the beamsplitter $\mc{B}^x$~\cite[Eqns.~(5)--(6)]{KMN+05}. This bound has implications for the reading protocols considered in \cite{WGTL12}. 

Let $\hat{H}$ denote the familiar $\hat{a}^\dag \hat{a}$ number observable
%for $\mc{B}_{\mc{X}}$
and let $N_S \in [0,\infty)$. The energy-constrained reading capacity
$\mc{C}(\mc{B}_{\mc{X}},\hat{H}, N_S)$
of a beamsplitter memory cell $\mc{B}_{\mc{X}}$ is defined in the obvious way, such that the average energy of the input to each call of the memory is bounded from above by $N_S \geq 0$. This definition implies that the function we optimize in the  capacity upper bound has the following constraint: for any input ensemble $\{p_X(x),\rho^x_{R{B'}}\}$:
\begin{equation}
\Tr\left\{\hat{H}\int p_X(x)\rho^x_{{B'}}\right\}\leq N_{S}.
\end{equation}
Since the energy of the output state of $\mc{B}^x$ does not depend on the phase $\phi$, we can drop the dependence of $x$ on $\phi$ and take $x=\eta$ for our discussion. 
For a memory cell $\mc{B}_{\mc{X}}$, the energy of the output state is constrained as
\begin{align}
& \Tr\left\{\sum_{x\in\mc{X}}p_X(x)\mc{B}^x(\rho^x_{B'})\hat{H}\right\} \notag \\
&=\sum_{x\in\mc{X}}p_X(x)\Tr\left\{\mc{B}^x(\rho^x_{B'})\hat{H}\right\}\\
&=\sum_{x\in\mc{X}}p_X(x)\eta \Tr\left\{\rho^x_{B'}\hat{H}\right\}\\
&\leq N_S,
\end{align}
where the second equality holds because the transmissivity of each $\mc{B}^x$ is $\eta \in [0,1]$.

We can then state the following result:
\begin{corollary}
The energy-constrained reading capacity of a beamsplitter memory cell $\mc{B}_{\mc{X}}=\{\mc{B}^x\}_{x\in\mc{X}}$ is bounded from above as
\begin{equation}
\mc{C}(\mc{B}_{\mc{X}},\hat{H}, N_S)\leq 2 g(N_S),
\end{equation}
where $\theta^{N_S}$ is a thermal state~\eqref{eq:thermal-state} such that $\Tr\{ \hat{H} \theta^{N_S}\} = N_S$ and $g(y):=(y+1)\log(y+1)-y\log y$.
\end{corollary}

\begin{proof}
From a straightforward extension of  Theorem~\ref{thm:q-r-capacity}, which takes into account the energy constraint, we find that 
\begin{align}
& \mc{C}(\mc{S}_{\mc{X}},\hat{H}, N_S) \notag \\
&\leq \sup_{\{p_X(x),\rho^x_{R{B'}}\} :\atop  \mathbb{E}_X\{\Tr\{\hat{H} \rho^X_{{B'}}\}\} \leq N_S }I(X;B|R)_\omega - I(X;{B'}|R)_\rho\\
& \leq \sup_{\{p_X(x),\rho^x_{R{B'}}\} :  \atop\mathbb{E}_X\{\Tr\{\hat{H} \rho^X_{{B'}}\}\} \leq N_S }I(X;B|R)_\rho\\
&\leq \sup_{\{p_X(x),\rho^x_{R{B'}}\}
 : \atop \mathbb{E}_X\{\Tr\{\hat{H} \rho^X_{{B'}}\}\} \leq N_S} 2H(B)_\rho\\
&\leq 2 H(\theta^{N_S})\\
&= 2g(N_S).
\end{align}
The first inequality follows from the extension of Theorem~\ref{thm:q-r-capacity}. The second inequality follows from non-negativity of the conditional quantum mutual information. The third inequality follows from a standard entropy bound for the conditional quantum mutual information. 
The fourth inequality follows because the thermal state of mean energy $N_S$ has the maximum entropy under a fixed energy constraint (see, e.g., \cite{Carlen09}). The final equality follows because the observable $\hat{H}$ is the familiar $\hat{a}^\dag \hat{a}$ number observable, for which the entropy of its thermal state of mean photon number $N_S$ is given by $g(N_S)$.
\end{proof}

\begin{remark}
It follows  that
\begin{equation}
\mc{C}_{\textnormal{non-adaptive}}(\mc{B}_{\mc{X}},\hat{H}, N_S)\leq
2g(N_S)
\end{equation}
 because
\begin{equation}
\mc{C}_{\textnormal{non-adaptive}}(\mc{B}_{\mc{X}},\hat{H}, N_S)\leq  \mc{C}(\mc{B}_{\mc{X}},\hat{H}, N_S),
\end{equation}
 by the definition of the energy-constrained quantum reading capacity of a memory cell~$\mc{B}_{\mc{X}}$. 
\end{remark}

\section{Examples of environment-parametrized memory cells}
\label{sec:example}

In this section, we calculate the quantum reading capacities of several environment-parametrized memory cells, including a thermal memory cell, and a jointly covariant memory cell formed from a channel $\mc{N}$ and a group $G$ with respect to which $\mc{N}$ is covariant. Examples of such a jointly covariant memory cell include qudit erasure and depolarizing memory cells formed respectively from erasure and depolarizing channels. 

\subsection{Jointly covariant memory cell: $\mc{N}^{\textnormal{cov}}_G$}

In this section, we show that the quantum reading capacity of a memory cell $\mc{N}^{\textnormal{cov}}_G$  as given in Definition~\ref{def:N-G-cell} below is equal to the entanglement-assisted classical capacity of the underlying channel $\mc{N}$. Our result makes use of the fact that the entanglement-assisted classical capacity of a covariant channel $\mc{T}$ is equal to $I(R;B)_{\mc{T}(\Phi)}$ \cite{BSST99,BSST02}. Furthermore, we use this result to evaluate the quantum reading capacity of a qudit erasure memory cell (Definition~\ref{def:qudit-erasure}) and a qudit depolarizing memory cell (Definition~\ref{def:qudit-dep}). 

\begin{definition}[$\mc{N}^{\textnormal{cov}}_{G}$]\label{def:N-G-cell}
Let $\mc{N}$ be a covariant channel with respect to a group $G$ as in Definition~\ref{def:covariant}. We define the memory cell $\mc{N}^{\textnormal{cov}}_G$ as
\begin{equation}
\mc{N}^{\textnormal{cov}}_G=\left\{\mc{N}_{{B'}\to B}\circ\mc{U}^g_{B'}\right\}_{g\in G},
\end{equation}
where $\mc{U}_{B'}^g:=U_{B'}(g)(\cdot)U^\dag_{B'}(g)$.  It follows from~\eqref{eq:cov-condition} that 
\begin{equation}\label{eq:con-cov-2}
\mc{N}_{{B'}\to B}\circ\mc{U}^g_{B'}= \mc{V}^g_B\circ\mc{N}_{{B'}\to B},
\end{equation}
where $\mc{V}_B^g:=V_B(g)(\cdot)V^\dag_B(g)$. It also follows that $\mc{N}^{\textnormal{cov}}_G$ is a  jointly covariant memory cell.
\end{definition}

\begin{theorem}\label{thm:covariant-to-EA-cap}
The quantum reading capacity $\mc{C}(\mc{N}^\textnormal{cov}_G)$ of the jointly covariant memory cell $\mc{N}^{\textnormal{cov}}_G=\left\{\mc{N}_{{B'}\to B}\circ\mc{U}^g_{B'}\right\}_{g\in G}$, as in Definition~\ref{def:N-G-cell}, is equal to the entanglement-assisted classical capacity of~$\mc{N}$:
\begin{equation}
\mc{C}(\mc{N}^\textnormal{cov}_G) = I(R;B)_{\mc{N}(\Phi)},
\end{equation}
where $\mc{N}(\Phi):=\mc{N}_{{B'}\to B}(\Phi_{R{B'}})$ and $\Phi_{R{B'}}\in\mc{D}(\mc{H}_{R{B'}})$ is a maximally entangled state. 
\end{theorem}

\begin{proof}
Our proof consists of two parts: the converse part and the achievability part. We first show the converse part:
\begin{equation}
\mc{C}\(\mc{N}^\textnormal{cov}_G\) \leq I(R;B)_{\mc{N}(\Phi)}. 
\end{equation}
From Remark~\ref{rem:cov-cell-achievability}, we conclude that the quantum reading capacity of $\mc{N}^\textnormal{cov}_G$ is as follows:
\begin{equation}
\mc{C}\(\mc{N}^\textnormal{cov}_G\)=\max_{p_G}I(G;RB)_{\omega},
\end{equation}
where
\begin{equation}
\omega_{GRB}:=\sum_{g\in G}p_G(g)|g\>\<g|_G\otimes\omega^{g}_{RB},
\end{equation}
such that $\{|g\>\}_{g\in G}$ forms an orthonormal basis on $\mc{H}_{G}$ and 
\begin{equation}
\forall g\in G:\ \omega^g_{RB}=(\mc{N}_{{B'}\to B}\circ\mc{U}^g_{B'})(\Phi_{R{B'}}).
\end{equation}
Let us fix $p_G$. Then
\begin{align}
& I(G;RB)_{\omega} \notag \\
&=H\!\(\sum_{g\in G}p_G(g)\omega^g_{RB}\)
 -\sum_{g\in G}p_G(g)H(\omega^{g}_{RB})\\
&=H\!\(\sum_{g\in G}p_G(g)(\mc{V}^g_B\circ\mc{N}_{{B'}\to B})(\Phi_{R{B'}})\)\notag \\
& \qquad -\sum_{g\in G}p_G(g)H((\mc{V}^g_B\circ\mc{N}_{{B'}\to B})(\Phi_{R{B'}}))\label{eq:observe-cov-uni}\\
%&=H\(\sum_{g\in G}p_G(g)\mc{V}^g_B\circ\mc{N}_{{B'}\to B}(\Phi_{R{B'}})\)-H\(\mc{N}_{{B'}\to B}(\Phi_{R{B'}})\)\\
&=\sum_{g^\prime\in G}\frac{1}{|G|}H\!\(\sum_{g\in G}p_G(g)(\mc{V}^{g^\prime}_B\circ\mc{V}^g_B\circ\mc{N}_{{B'}\to B})(\Phi_{R{B'}})\)\notag\\
& \qquad -H(\mc{N}_{{B'}\to B}(\Phi_{R{B'}}))\\
&\leq H\!\(\frac{1}{|G|}\sum_{g,g^\prime\in G}p_G(g)(\mc{V}^{g^\prime}_B\circ\mc{V}^g_B\circ\mc{N}_{{B'}\to B})(\Phi_{R{B'}})\)\notag \\
& \qquad -H(\mc{N}_{{B'}\to B}(\Phi_{R{B'}}))\\
&=H\!\!\(\mc{N}_{{B'}\to B}\!\!\(\frac{1}{|G|}\sum_{g^\prime\in G}\mc{U}^{g^\prime}_{B'}\!\!\(\sum_{g\in G}p_G(g)\mc{U}^g_{B'}(\Phi_{R{B'}})\)\)\)\notag \\
& \qquad -H(\mc{N}_{{B'}\to B}(\Phi_{R{B'}}))\\
&=H\(\mc{N}_{{B'}\to B}(\pi_R\otimes\pi_{B'})\)-H(\mc{N}_{{B'}\to B}(\Phi_{R{B'}}))\\
&=H\(\pi_R\)+H\(\mc{N}_{{B'}\to B}(\pi_B)\)-H(\mc{N}_{{B'}\to B}(\Phi_{R{B'}}))\\
&=I(R;B)_{\mc{N}(\Phi)}.
\end{align}
The second equality follows from~\eqref{eq:con-cov-2}. The third equality follows because entropy is invariant with respect to unitary or isometric channels. The first inequality follows from the concavity of entropy. The fourth equality follows from~\eqref{eq:con-cov-2}. The fifth equality follows from Definition~\ref{def:covariant}. The sixth equality follows because entropy is additive for product states. 
Since the above upper bound holds for any $p_G$, it follows that
\begin{equation}\label{eq:cov-con}
\mc{C}\(\mc{N}^\textnormal{cov}_G\)=\max_{p_G}I(G;RB)_{\omega}\leq I(R;B)_{\mc{N}(\Phi)}.
\end{equation} 

To prove the achievability part, we take $p_G$ to be a uniform distribution, i.e., $p_G\sim \frac{1}{|G|}$. Putting $p_G\sim\frac{1}{|G|}$ in \eqref{eq:observe-cov-uni}, we find the following lower bound:
\begin{equation}\label{eq:cov-ach}
\mc{C}\(\mc{N}^\textnormal{cov}_G\) \geq I(G;RB)_\omega = I(R;B)_{\mc{N}(\Phi)}.
\end{equation}
Thus, from \eqref{eq:cov-con} and \eqref{eq:cov-ach} we conclude the statement of the theorem.
\end{proof}

\bigskip 
Now we state two corollaries, which are direct consequences of the above theorem. These corollaries establish the quantum reading capacities for jointly covariant memory cells formed from the erasure channel and depolarizing channel with respect to the Heisenberg--Weyl group $\mathbf{H}$, as discussed below. 

Let us first introduce some basic notations and definitions related to qudit systems. A system represented with a $d$-dimensional Hilbert space is called a qu$d$it system. 

Let $J_{B'}=\{|j\>_{B'}\}_{j\in \{0,\ldots,d-1\}}$ be a computational orthonormal basis of $\mc{H}_{B'}$ such that $\dim(\mc{H}_{B'})=d$. There exists a unitary operator called the \textit{cyclic shift operator} $X(k)$ that acts on the orthonormal states as follows:
\begin{equation}
\forall |j\>_{B'}\in J_{B'}:\ \ X(k)|j\>=|k\oplus j\>,
\end{equation}
where $\oplus$ is a cyclic addition operator, i.e., $k\oplus j:= (k+j)\ \textnormal{mod}\ d$. There also exists another unitary operator called the \textit{phase operator} $Z(l)$ that acts on the qudit computational basis states as
\begin{equation}
\forall |j\>_{B'}\in J_{B'}:\ \ Z(l)|j\>=\exp\(\frac{\iota 2\pi lj}{d}\)|j\>. 
\end{equation}
The $d^2$ operators $\{X(k)Z(l)\}_{k,l\in\{0,\ldots,d-1\}}$ are known as the Heisenberg--Weyl operators. Let $\sigma(k,l):=X(k)Z(l)$. 
The maximally entangled state $\Phi_{R{B'}}$ of qudit systems $R{B'}$ is given as
\begin{equation}
|\Phi\>_{R{B'}}:=\frac{1}{\sqrt{d}}\sum_{j=0}^{d-1}|j\>_R|j\>_{B'},
\end{equation}
and we define 
\begin{equation}
|\Phi^{k,l}\>_{R{B'}}:=(I_R\otimes\sigma^{k,l}_{B'})|\Phi\>_{R{B'}}.
\end{equation}
The $d^2$ states $\{|\Phi^{k,l}\>_{R{B'}}\}_{k,l\in\{0,\ldots,d-1\}}$ form a complete, orthonormal basis:
\begin{align}
\<\Phi^{k_1,l_1}|\Phi^{k_2,l_2}\>&=\delta_{k_1,k_2}\delta_{l_1,l_2},\\
\sum_{k,l=0}^{d-1}|\Phi^{k,l}\>\<\Phi^{k,l}|_{R{B'}}&=I_{R{B'}}.
\end{align}

Let $\mc{W}$ be a discrete set such that $|\mc{W}|=d^2$. There exists a one-to-one mapping $\{(k,l)\}_{k,l\in\{0,d-1\}}\leftrightarrow \{w\}_{w\in\mc{W}}$. For example, we can use the following map: $w=k+d\cdot l$ for $\mc{W}=\{0,\ldots,d^2-1\}$. This allows us to define $\sigma^w:=\sigma(k,l)$ and $\Phi^w_{R{B'}}:=\Phi^{k,l}
_{R{B'}}$.  Let the set of $d^2$ Heisenberg--Weyl operators be denoted as
\begin{equation}\label{eq:HW-op}
\mathbf{H}:=\{ \sigma^w\}_{w\in\mc{W}}=\{X(k)Z(l)\}_{k,l\in\{0,\ldots,d-1\}},
\end{equation}
and we refer to $\mathbf{H}$  as the Heisenberg--Weyl group. 

\begin{definition}[Qudit erasure memory cell]\label{def:qudit-erasure}
The qudit erasure memory cell $\mc{Q}^q_{\mc{X}}=\left\{\mc{Q}^{q,x}_{{B'}\to B}\right\}_{x\in\mc{X}}$, $\vert\mc{X}\vert=d^2$, consists of the following qudit channels:
\begin{equation}
\mc{Q}^{q,x}(\cdot)=\mc{Q}^q(\sigma^x(\cdot)\(\sigma^x\)^\dag) ,
\end{equation}
where $\mc{Q}^q$ is a qudit erasure channel \cite{GBP97}:
\begin{equation}
\mc{Q}^q(\rho_{B'})=(1-q)\rho+q|e\>\<e|
\end{equation}
such that $q\in [0,1]$, $\dim(\mc{H}_{B'})=d$, $|e\>\<e|$ is some state orthogonal to the support of any input state $\rho$, and
$\forall x\in\mc{X}: \sigma^x\in\mathbf{H}$ are the Heisenberg--Weyl operators as given in \eqref{eq:HW-op}. Observe that $\mc{Q}^q_{\mc{X}}$ is jointly covariant with respect to the Heisenberg--Weyl group $\mathbf{H}$ because the qudit erasure channel $\mc{Q}^q$ is covariant with respect to $\mathbf{H}$.
\end{definition}

\begin{definition}[Qudit depolarizing memory cell]\label{def:qudit-dep}
The qudit depolarizing memory cell $\mc{D}^q_{\mc{X}}=\left\{\mc{D}^{q,x}_{{B'}\to B}\right\}_{x\in\mc{X}}$, $|\mc{X}|=d^2$, consists of qudit channels
\begin{equation}
\mc{D}^{q,x}(\cdot)=\mc{D}^q\(\sigma^x(\cdot)\(\sigma^x\)^\dag\)
\end{equation}
where $\mc{D}^q$ is a qudit depolarizing channel:
\begin{equation}
\mc{D}^q(\rho)=(1-q)\rho+q\pi,
\end{equation}
where $q\in[0,\frac{d^2}{d^2-1}]$, $\dim(\mc{H}_{B'})=d$ and 
$\forall x\in\mc{X}: \sigma^x\in\mathbf{H}$
are the Heisenberg--Weyl operators as given in \eqref{eq:HW-op}.  
Observe that $\mc{D}^q_{\mc{X}}$ is jointly covariant with respect to the Heisenberg--Weyl group $\mathbf{H}$ because the qudit depolarizing channel $\mc{D}^q$ is covariant with respect to $\mathbf{H}$.
\end{definition}

As a consequence of Theorem~\ref{thm:covariant-to-EA-cap}, we immediately find the quantum reading capacities of the above memory cells:

\begin{corollary}
The quantum reading capacity $\mc{C}(\mc{Q}^q_{\mc{X}})$ of the qudit erasure memory cell $\mc{Q}^q_{\mc{X}}$ (Definition~\ref{def:qudit-erasure}) is equal to the entanglement-assisted classical capacity of the erasure channel $\mc{Q}^q$ \cite{BSST99}:
\begin{equation}
\mc{C}(\mc{Q}^q_{\mc{X}})=2(1-q)\log_2d.
\end{equation}
\end{corollary}

\begin{corollary}
The quantum reading capacity $\mc{C}(\mc{D}^q_{\mc{X}})$ of the qudit depolarizing memory cell $\mc{D}^q_{\mc{X}}$ (Definition~\ref{def:qudit-dep}) is equal to the entanglement-assisted classical capacity of the depolarizing channel $\mc{D}^q$ \cite{BSST99}:
\begin{multline}
\mc{C}(\mc{D}^q_{\mc{X}})=2\log_2d+\(1-q+\frac{q}{d^2}\)\log_2\!\(1-q+\frac{q}{d^2}\)\\
+(d^2-1)\frac{q}{d^2}\log_2\!\(\frac{q}{d^2}\).
\end{multline}
\end{corollary}

\subsection{A thermal memory cell}

Let us consider an example of a thermal memory cell $\mc{E}_{\mc{X},\eta} = \{ \mc{E}^{x,\eta}\}_x$, which is an environment-parametrized memory cell consisting of thermal channels $\mc{E}^{x,\eta}$ with known transmissivity parameter $\eta\in[0,1]$ and unknown excess noise $x$ \cite{TW16}. Let $\hat{a},\hat{b},\hat{e},\hat{e}^\prime$ be the respective field-mode annihilation operators for Bob's input, Bob's output, the environment's input, and the environment's output of these channels. The interaction channel in this case is a fixed unitary $U_{{B'}E\to BE^\prime}$ corresponding to a beamsplitter interaction, defined from the following Heisenberg input-output relations:
\begin{align}
\hat{b}&=\sqrt{\eta}\hat{a}+\sqrt{1-\eta}\hat{e},\\
\hat{e}^\prime&=-\sqrt{1-\eta}\hat{a}+\sqrt{\eta}\hat{e}.
\end{align}
The environmental mode $\hat{e}$ of a thermal channel $\mc{E}^{x,\eta}$ is prepared in a thermal state $\theta^x:=\theta(N_B=x)$ of mean photon number $N_B\geq 0$: 
\begin{equation}\label{eq:thermal-state}
\theta(N_B):=\frac{1}{N_B+1}\sum_{k=0}^\infty\(\frac{N_B}{N_B+1}\)^k|k\>\<k|,
\end{equation}
where $\{|k\>\}_{k\in\mathbb{N}}$ is the orthonormal, photonic number-state basis. Parameter $x$ is the excess noise of the thermal channel $\mc{E}^{x,\eta}$. It can be noted that for $x=0$, $\theta^x$ reduces to a vacuum state and the channel $\mc{E}^{x,\eta}$ is called the pure-loss channel.  

\begin{proposition}
The quantum reading capacity $\mc{C}(\mc{E}_{\mc{X},\eta})$ of the thermal memory cell $\mc{E}_{\mc{X},\eta} = \{ \mc{E}^{x,\eta}\}_x$ (as described above) is equal to
\begin{equation}
\mc{C}(\mc{E}_{\mc{X},\eta}) = \sup_{p_X}
\left[H(\overline{\theta}) - \int dx \ p_X(x) H(\theta^x)\right],
\end{equation}
where $p_X$ is a probability distribution for the parameter $x$ and $\overline{\theta} =
\int dx \ p_X(x) \theta^x$.
\end{proposition}

\begin{proof}
We begin by proving the achievability part, which corresponds to the inequality
\begin{equation}
\mc{C}(\mc{E}_{\mc{X},\eta}) \geq I(X;E)_\theta,
\end{equation}
where $\theta_{XE}=\int dx \ p_X(x)|x\>\<x|_X\otimes\theta^x_E$. The main idea for the achievability part builds on the results of \cite[Eqns.~(38)--(48)]{TW16}.

The two-mode squeezed vacuum state is equivalent to a purification of the thermal state in \eqref{eq:thermal-state} and is defined as
\begin{equation}
\left|\phi^\textnormal{TMS}(N_S)\right>_{R{B'}}:=\frac{1}{\sqrt{N_S+1}}\sum_{k=0}^\infty\[\frac{N_S}{N_S+1}\]^{\frac{k}{2}}|k\>_R|k\>_{B'}.
\end{equation}
When sending the ${B'}$ system of this state through the channel $\mc{E}^{x,\eta}_{ {B'}\to B}$, the output state is as follows:
\begin{align}
 \omega^{x,\eta}_{RB}(N_S) & :=
 (\id_R\otimes\mc{E}^{x,\eta}_{ {B'}\to B})\(\phi^\textnormal{TMS}_{R{B'}}(N_S)\)\label{eq:omega-x}\\
&=\Tr_{E^\prime}\{\mathcal{U}_{{B'}E\to BE^\prime}\(\phi^\textnormal{TMS}_{R{B'}}(N_S)\otimes\theta^x_E\)\}\label{eq:theta-x-thermal},
\end{align}
where $\mathcal{U}_{{B'}E\to BE^\prime}(\cdot) := U_{{B'}E\to BE^\prime} (\cdot) U_{{B'}E\to BE^\prime}^\dag$,
and the average output state is as follows, when the channel $\mc{E}^{x,\eta}_{ {B'}\to B}$ being applied is chosen with probability
$p_X(x)$:
\begin{align}
& \sum_{x\in\mc{X}}p_X(x)\omega^{x,\eta}_{RB}(N_S)\notag \\
&=\sum_{x\in\mc{X}}p_X(x)\Tr_{E^\prime}\left\{\mathcal{U}_{{B'}E\to BE^\prime}\(\phi_{R{B'}}(N_S)\otimes\theta^x_E\)\right\}\\
&=\Tr_{E^\prime}\left\{\mathcal{U}_{{B'}E\to BE^\prime}\(\phi_{R{B'}}(N_S)\otimes\sum_{x\in\mc{X}}p_X(x)\theta^x_E\)\right\}.
\end{align}
Let us define the following classical--quantum state:
\begin{equation}
\omega^{\eta}_{XRB}(N_S)=\sum_{x\in\mc{X}}p_X(x)|x\>\<x|_X\otimes\omega^{x,\eta}_{RB},
\end{equation}
and consider that
\begin{multline}
I(X;RB)_{\omega^{\eta}(N_S)}=\\
\sum_{x\in\mc{X}}p_X(x)D\!\(\omega^{x,\eta}_{RB}(N_S)\middle\Vert\sum_{x\in\mc{X}}p_X(x)\omega^{x,\eta}_{RB}(N_S)\) .
\end{multline}

The Wigner characteristic function covariance matrix \cite{adesso14} for $\omega^{x,\eta}_{RB}(N_S)$ in \eqref{eq:omega-x} is as follows:
\begin{equation}
V_{\omega^{x,\eta}(N_S)}=\begin{bmatrix}
a & c & 0 & 0 \\
c & b & 0 & 0 \\
0 & 0 & a & -c \\
0 & 0 & -c & b \\
\end{bmatrix},
\end{equation}
where
\begin{align}
a & =\eta N_S+\(1-\eta\)x+\frac{1}{2},\\
b& =N_S+\frac{1}{2},\\
c&=\sqrt{\eta N_S(N_S+1)} . 
\end{align}

Let us consider the following symplectic transformation \cite{TW16}:
\begin{equation}
S^{\eta}(N_S)=\begin{bmatrix}
\gamma_+ & -\gamma_- & 0 & 0 \\
-\gamma_- & \gamma_+ & 0 & 0 \\
0 & 0 & \gamma_+ & \gamma_- \\
0 & 0 & \gamma_- & \gamma_+ \\
\end{bmatrix} ,
\end{equation}
where
\begin{align}
\gamma_+=\sqrt{\frac{1+N_S}{1+(1-\eta)N_S}},\qquad
\gamma_-=\sqrt{\frac{\eta N_S}{1+(1-\eta)N_S}}\ .
\end{align}

The action of the symplectic matrix $S^{\eta}(N_S)$ on the covariance matrix $V_{\omega^{x,\eta}(N_S)}$ gives
\begin{align}
\hat{V}_{\omega^{x,\eta}(N_S)}& :=S^{\eta}(N_S)V_{\omega^{x,\eta}(N_S)}\(S^{\eta}(N_S)\)^\textnormal{T}\\
& =\begin{bmatrix}
a_s & -c_s & 0 & 0 \\
-c_s & b_s & 0 & 0 \\
0 & 0 & a_s & c_s \\
0 & 0 & c_s & b_s \\
\end{bmatrix},
\end{align}
where
\begin{align}
a_s&=x+\frac{1}{2}+\mc{O}\(\frac{1}{N_S}\),\\
b_s&=\(1-\eta\)N_S+\eta x+\frac{1}{2}+\mc{O}\(\frac{1}{N_S}\),\\
c_s&=\sqrt{\eta}x+\mc{O}\(\frac{1}{N_S}\)\ . 
\end{align}
Thus, by applying this transformation to $\omega^{x,\eta}(N_S)$ and tracing out the second mode, we are left with a state that becomes indistinguishable from a thermal state of mean photon number $x$ in the limit as $N_S \to \infty$. Note that this occurs independent of the value of the transmissivity~$\eta$.

The symplectic transformation $S^{\eta}(N_S)$ can be realized by a two-mode squeezer, which corresponds to a unitary transformation acting on the tensor-product Hilbert space. Letting the unitary transformation be of the form $W_{RB\to EB}$, then $\hat{V}_{\omega^{x,\eta}(N_S)}$ represents the covariance matrix of the state $\omega^{x,\eta}_{EB}(N_S)$. 

We use the formula for fidelity between two thermal states \cite[Equation 34]{TW16} and the relation between trace norm and fidelity \cite[Theorem 9.3.1]{W15book} to conclude that
\begin{multline}
\lim_{N_S\to\infty}\left\Vert \omega^{x,\eta}_E(N_S)-\theta^x_E\right\Vert_1\\
\leq \lim_{N_S\to\infty}\sqrt{1-F\(\omega^{x}_E(N_S),\theta^x_E\)}=0.
\end{multline}
From the convexity of trace norm, we find that
\begin{multline}
\left\Vert \sum_{x\in\mc{X}}p_X(x)\omega^x_E(N_S)-\sum_{x\in\mc{X}}p_X(x)\theta^x_E\right\Vert_1
\\
\leq \sum_{x\in\mc{X}}p_X(x)\left\Vert \omega^{x}_E(N_S)-\theta^x_E\right\Vert_1,
\end{multline}
which in turn implies that
\begin{align}
 \lim_{N_S\to\infty}\left\Vert \sum_{x\in\mc{X}}p_X(x)\omega^x_E(N_S)-\sum_{x\in\mc{X}}p_X(x)\theta^x_E\right\Vert_1=0.
\end{align}

Invoking the result of \cite[Equation 28]{TW16} and the lower semi-continuity of relative entropy, we get that
\begin{multline}
\lim_{N_S\to\infty} D\(\omega^{x,\eta}_{RB}(N_S)\left\Vert\sum_{x\in\mc{X}}p_X(x)\omega^{x,\eta}_{RB}(N_S)\)\right.\\
=D\(\theta^x_{E}\left\Vert \sum_{x\in\mc{X}}p_X(x)\theta^x_E\)\right..
\end{multline}

Thus, from the above relations, we find that
\begin{equation}
\lim_{N_S\to \infty}I(X;RB)_{\omega^{\eta}(N_S)}=I(X;E)_{\theta},
\end{equation}
where 
\begin{equation}
\theta_{XE}=\sum_{x\in\mc{X}}p_X(x)|x\>\<x|_X\otimes\theta^x_E ,
\end{equation}
for $\theta^x_E$ defined in Equation~\eqref{eq:theta-x-thermal}. This shows that $I(X;E)_{\theta}$ is an achievable rate for any $p_X$. 

The converse part of the proof, which corresponds to the inequality 
\begin{equation}
\mc{C}(\mc{E}_{\mc{X},\eta})\leq \max_{p_X}I(X;E)_\theta,
\end{equation}
follows directly from Theorem~\ref{thm:env-cell-qrc}.
\end{proof}

\section{Zero-error quantum reading capacity}\label{sec:zero-error}

In an $(n,R,\varepsilon)$ quantum reading protocol (Definition~\ref{def:QR}) for a memory cell $\mc{S}_{X}=\{\mc{M}^x_{{B'}\to B}\}_{x\in\mc{X}}$, one can demand the error probability to vanish, i.e., $\varepsilon=0$. In this section, we define zero-error quantum reading protocols and the zero-error quantum reading capacity for any memory cell. We provide an explicit example of a memory cell for which a quantum reading protocol using an adaptive strategy has a clear advantage over a quantum reading protocol that uses a non-adaptive strategy. 

\begin{definition}[Zero-error quantum reading protocol]
A zero-error quantum reading protocol
of a memory cell $\mc{S}_{\mc{X}}$ is a particular 
$(n,R,\varepsilon)$ quantum reading protocol for which $\varepsilon=0$.
\end{definition}

\begin{definition}[Zero-error quantum reading capacity]
The zero-error quantum reading capacity
$\mc{Z}(\mc{S}_{\mc{X}})$ of a memory cell $\mc{S}_{\mc{X}}$ is defined as the largest rate $R$ such that there exists a zero-error reading protocol. 
\end{definition}

A zero-error non-adaptive quantum reading protocol of a memory cell is a special case of a zero-error quantum reading protocol in which the reader uses a non-adaptive strategy to decode the message. 
 
\subsection{Advantage of an adaptive strategy over a non-adaptive strategy}

In this section, we employ the main example from
\cite{HHLW10} to  illustrate the advantage of an adaptive zero-error quantum reading protocol over a non-adaptive zero-error quantum reading protocol. 

Let us consider a memory cell $\mc{B}_{\mc{X}}=\{\mc{M}_{{B'}\to B}^x\}_{x\in\mc{X}}$, $\mc{X}=\{1,2\}$, consisting of the following quantum channels that map two qubits to a single qubit,  acting as
\begin{equation}
\mc{M}^x(\cdot)=\sum_{j=1}^5A^x_j(\cdot)\(A^x_j\)^\dag, \ x\in\mc{X} ,
\end{equation}
where
\begin{align}
A^1_1 & =|0\>\<00|,
  & A^1_2 & =|0\>\<01|,\\
     A^1_3 & =|0\>\<10|,
    &  A^1_4 & =\frac{1}{\sqrt{2}}|0\>\<11|,\\
      A^1_5 & =\frac{1}{\sqrt{2}}|1\>\<11|, 
& A^2_1 & =|+\>\<00|,\\
A^2_2 & =|+\>\<01|,
 & A^2_3 & =|1\>\<1+|,\\ 
 A^2_4 & =\frac{1}{\sqrt{2}}|0\>\<1-|,
  & A^2_5 & =\frac{1}{\sqrt{2}}|1\>\<1-|,
\end{align}
and the standard bases for the channel inputs and outputs are $\{|00\>,|01\>,|10\>,|11\>\}$ and $\{|0\>,|1\>\}$, respectively. 
 
It follows from \cite{HHLW10} that it is possible to discriminate perfectly these two channels using an adaptive strategy that makes two calls to the unknown channel
$\mc{M}^x$. This implies that the encoder can encode  two classical messages (one bit) into two uses of the quantum channels from $\mc{B}_{\mc{X}}$ such that Bob can perfectly read the message, i.e., with zero error. Thus, we can conclude that the zero-error quantum reading capacity
of $\mc{B}_{\mc{X}}$
is  bounded
from below by $\frac{1}{2}$ (one bit per two channel uses).

Closely following the arguments of \cite[Section 4]{HHLW10}, we can show that non-adaptive strategies can never realize perfect discrimination of the sequences $\mc{M}_{{B'}^n\to B^n}^{x^n}$ and $\mc{M}_{{B'}^n\to B^n}^{y^n}$, for any finite number $n$ of channel uses if $x^n\neq y^n$.  Equivalently, for $x^n\neq y^n$
\begin{equation}
 \Vert \mc{M}_{{B'}^n\to B^n}^{x^n}-\mc{M}_{{B'}^n\to B^n}^{y^n}\Vert_{\diamond}< 2 \label{eq:ch-discrimination}
\end{equation} 
for all $ n\in\mathbb{N}$,
where $\Vert \cdot\Vert_{\diamond}$ is the diamond norm (defined in \cite[Equation 1]{HHLW10}). Thus, the zero-error non-adaptive quantum reading capacity of $\mc{B}_{\mc{X}}$ is equal to zero. 

To prove the above claim, we proceed with a proof by contradiction along the lines of that given in \cite[Section 4]{HHLW10}. We need to show that for any finite $n\in\mathbb{N}$, if $x^n\neq y^n$, then there does not exist any state $\sigma_{R{B'}^n}$ such that the two sequences $\mc{M}_{{B'}^n\to B^n}^{x^n}$ and $\mc{M}_{{B'}^n\to B^n}^{y^n}$ can be perfectly discriminated. Note that perfect discrimination is possible if and only if  there exists a state $\sigma_{R{B'}^n}$ such that 
\begin{equation}
\label{eq:ch-dis}
\Tr\left\{ \mc{M}_{{B'}^n\to B^n}^{x^n}(\sigma_{R{B'}^n})\mc{M}_{{B'}^n\to B^n}^{y^n} (\sigma_{R{B'}^n})\right\}=0,
\end{equation}
which is the same as the condition
\begin{equation}
 \mc{M}_{{B'}^n\to B^n}^{x^n}(\sigma_{R{B'}^n})\mc{M}_{{B'}^n\to B^n}^{y^n} (\sigma_{R{B'}^n})=0,
\end{equation}
and in turn the same as
\begin{equation}
\Tr\left\{ \sqrt{\mc{M}_{{B'}^n\to B^n}^{x^n}(\sigma_{R{B'}^n})}\sqrt{\mc{M}_{{B'}^n\to B^n}^{y^n} (\sigma_{R{B'}^n})}\right\}=0.
\end{equation}
Let us suppose that there exists a state $\sigma_{R{B'}^n}$ such that \eqref{eq:ch-dis} holds. The data processing inequality for the quantity above \cite{P85,P86} then implies that \eqref{eq:ch-dis} holds  for some pure state $\psi_{R{B'}^n}$. Then, by carefully following the steps from \cite[Section 4]{HHLW10}, \eqref{eq:ch-dis} implies that 
for any set of complex coefficients $\{\alpha^{x,y}_{j,k}\in\mathbb{C}:1\leq j,k\leq 5,\ x,y\in\mc{X}\}$  
\begin{multline}
\langle\psi|_{R{B'}^n} \Big[ I_R \otimes \sum_{1\leq j,k\leq 5\,:\,i\in[n]}\alpha^{x_1,y_1}_{j_1,k_1}\cdots\ \alpha^{x_n,y_n}_{j_n,k_n}\({B'}^{y_1}_{j_1}\)^\dag \times \\
{B'}^{x_1}_{k_1}\otimes\cdots\otimes\({B'}^{y_n}_{j_n}\)^\dag {B'}^{x_n}_{k_n}\Big]|\psi\rangle_{R{B'}^n}= 0.\label{eq:for-a-contra}
\end{multline}
%The intermediate steps are same as those discussed in \cite[Section 4]{HHLW10}.
Let us choose the coefficients $\{\alpha^{x,y}_{j,k}\in\mathbb{C}:1\leq j,k\leq 5,\ x,y\in\mc{X}\}$ as follows: 
\begin{align}
\alpha^{x,y}_{1,1} & =\alpha^{x,y}_{2,2}=\sqrt{2}, \\\alpha^{x,y}_{3,5} & =\alpha^{x,y}_{4,3}=1, \\
\alpha^{x,y}_{4,4} & =-2\sqrt{2}, \\
\text{ otherwise }   \alpha^{x,y}_{j,k} & =0,
\end{align}
for\ $x\neq y$
and 
$\alpha^{x,y}_{j,k}=\delta_{j,k}$
for $x=y$,
where, if $j=k$ then $\delta_{j,k}=1$, else $\delta_{j,k}=0$.
For the above choice of the coefficients, it follows  that 
\begin{multline}
I_R \otimes \sum_{1\leq j,k\leq 5\,:\,i\in[n]}\alpha^{x_1,y_1}_{j_1,k_1}\cdots\ \alpha^{x_n,y_n}_{j_n,k_n}\(A^{y_1}_{j_1}\)^\dag \times \\
A^{x_1}_{k_1}\otimes\cdots\otimes\(A^{y_n}_{j_n}\)^\dag A^{x_n}_{k_n}\\
=I_R \otimes P^{x_1,y_1}\otimes\cdots\otimes P^{x_n,y_n}
\end{multline}  
where  for  $i\in[n]$
\begin{equation}
\ P^{x_i,y_i}= \left\{ 
\begin{tabular}{c c}
$P>0$, & $x_i\neq y_i$\\
$I>0$, &  otherwise,
\end{tabular} 
\right.
\end{equation}
and $P=|00\>\<00|+|01\>\<01|+|11\>\<11|+|1-\>\<1-|$. Observe that the operator
$I_R \otimes P^{x_1,y_1}\otimes\cdots\otimes P^{x_n,y_n}$ is positive definite.
This means that there cannot exist any state that satisfies \eqref{eq:for-a-contra}, and as a consequence \eqref{eq:ch-dis}, and this concludes the proof.

From the above discussion, we conclude that the zero-error quantum reading capacity of the memory cell $\mc{B}_{\mc{X}}$ is bounded from below by $\frac{1}{2}$ whereas the zero-error non-adaptive quantum reading capacity is equal to zero.

\section{Conclusion}
\label{sec:conclusion}

In this paper, we have provided the most general and natural definitions for quantum reading protocols and quantum reading capacities. We have introduced environment-parametrized memory cells for quantum reading, which are sets of quantum channels obeying certain symmetries. We have provided upper bounds on the quantum reading capacity and the non-adaptive quantum reading capacity of an arbitrary memory cell. We have also provided strong converse  and second-order  bounds on quantum reading capacities of environment-parametrized memory cells. We have calculated quantum reading capacities for a thermal memory cell, a qudit erasure memory cell, and a qudit depolarizing memory cell. Finally, we have shown the advantage of an adaptive strategy over a non-adaptive strategy in the context of zero-error quantum reading capacity of a memory cell.

We note that it is possible to use the methods developed here to obtain bounds on the quantum reading capacities of memory cells based on amplifying bosonic channels, in the same spirit as the results of a thermal memory cell (the argument follows from \cite{TW16}).

A natural question following from the developments in this paper is whether there exists a memory cell for which the quantum reading capacity is larger than what one could achieve by using a non-adaptive strategy. As discussed above, we have provided a positive answer to this question in the setting of zero error. However, the question remains open for the case of Shannon-theoretic capacity (i.e., with arbitrarily small error). We suspect that this question will have a positive answer, and we strongly suspect it will be the case in the setting of non-asymptotic capacity, our latter suspicion being due to the fact that
feedback is known to help in non-asymptotic settings for communication
(see, e.g., \cite{PPV11feedback}). We leave the investigation of this question for future work.
%. In the classical setting, it is known that if a pair of discrete memoryless (classical) channels cannot be perfectly distinguished with one evaluation, then they cannot be perfectly distinguished with any finite number of evaluations. However, in general, feedback (adaptive strategy) can help in improving the trade-off between reliability and delay of classical channels at rates below capacity (see, e.g., \cite{CYT09}). As classical channels are special class of quantum channels, we can infer to have advantage in discriminating among quantum channels employing adaptive strategy over non-adaptive strategy with non-zero error when finite number of uses of channels is allowed. %Even if it is possible for
%a pair of quantum channels to be discriminated perfectly when multiple evaluations are available, but not in the single evaluation case, adaptive strategy may still provide some advantage when only finite evaluations are allowed. 
%It would be interesting for future work to determine error exponents of the quantum reading rate of an arbitrary memory cell.

%Recently, a private reading protocol was introduced in \cite{DBW17} for a task of secure reading against passive eavesdropper. Bounds on the private reading capacities were derived exploiting the fact that reading protocol can be understood as particular instance of information processing task from bipartite interactions \cite{DBW17}. 

\bigskip
\textbf{Acknowledgements.}
We thank David Ding, Saikat Guha, and Andreas Winter for insightful discussions. We also thank the anonymous referees for feedback that helped to improve our paper. SD acknowledges support from
the LSU Graduate School Economic Development Assistantship.
MMW acknowledges support from the Office of Naval Research and the National Science Foundation.
%\end{acknowledgments}

\bibliographystyle{unsrt}
\bibliography{qr}

\end{document}